\documentclass[floatfix,twocolumn,aps,superscriptaddress,nofootinbib,pra,10pt]{revtex4-2}

\usepackage{bm}
\usepackage{amsthm}
\usepackage{amsfonts,amssymb,amsmath}
\usepackage{mathtools}
\usepackage{physics}

\usepackage{xcolor}
\usepackage{graphicx}
\usepackage[T1]{fontenc}
\usepackage{bbold}
\usepackage{url}
\usepackage{float}
\usepackage[utf8]{inputenc}
\usepackage{csquotes}
\usepackage[colorlinks=true]{hyperref}
\hypersetup{
    allcolors  = {blue},
    breaklinks = {true},
}
\usepackage[english]{babel}
\usepackage[normalem]{ulem}
\usepackage{xfrac}
\usepackage{stmaryrd}
\usepackage{microtype}
\allowdisplaybreaks

\usepackage{tikz}
\usetikzlibrary{shapes, shapes.geometric, shapes.symbols, shapes.arrows, shapes.multipart, shapes.callouts, shapes.misc,decorations.pathmorphing}
\tikzset{snake it/.style={decorate, decoration=snake}}

\theoremstyle{definition}
\newtheorem{dfn/}{Definition}

\newenvironment{dfn}
  {%
   \pushQED{\qed}\begin{dfn/}}
  {\popQED\end{dfn/}}

\newtheorem*{dem/}{Proof}
\newenvironment{dem}
  {%
   \pushQED{\qed}\begin{dem/}}
  {\popQED\end{dem/}}
  
\newtheorem{exmx}{Example}
\newenvironment{exm}
  {%
  \pushQED{\qed}\exmx}
  {\popQED\endexmx}

\theoremstyle{plain}
\newtheorem{theorem}{Theorem}

\newtheorem{lemma}{Lemma}
\newtheorem{coro}{Corollary}

\begin{document}
\title{
Using a resource theoretic perspective to witness and engineer quantum generalized contextuality for prepare-and-measure scenarios
}
\author{Rafael Wagner}
\email{rafael.wagner@inl.int}
\affiliation{University of S\~{a}o Paulo, Institute of Physics, S\~{a}o Paulo, SP, Brazil}
\affiliation{International Iberian Nanotechnology Laboratory (INL),
Av. Mestre Jos\'{e} Veiga, 4715-330 Braga, Portugal}
\affiliation{Centro de F\'{i}sica, Universidade do Minho, Braga 4710-057, Portugal}
\author{Roberto D. Baldijão}
\affiliation{Instituto de F\'{i}sica Gleb Wataghin, Universidade Estadual de Campinas, Campinas, SP, Brazil}
\affiliation{Institute for Quantum Optics and Quantum Information, Austrian Academy of Sciences, Boltzmanngasse 3, A-1090 Vienna, Austria}
\author{Alisson Tezzin}
\affiliation{University of S\~{a}o Paulo, Institute of Physics, S\~{a}o Paulo, SP, Brazil}
\author{Bárbara Amaral}
\affiliation{University of S\~{a}o Paulo, Institute of Physics, S\~{a}o Paulo, SP, Brazil}
\date{\today} 

\begin{abstract}

We employ the resource theory of generalized contextuality as a tool for analyzing the structure of  prepare-and-measure scenarios. We argue that this framework simplifies proofs of quantum contextuality in complex scenarios and strengthens existing arguments regarding robustness of experimental implementations. As a case study, we demonstrate quantum contextuality associated with any nontrivial noncontextuality inequality for a class of useful scenarios by noticing a connection between the resource theory and measurement simulability. Additionally, we expose a formal composition rule  that allows engineering complex scenarios from simpler ones. This approach provides insights into the noncontextual polytope structure for complex scenarios and facilitates the identification of possible quantum violations of noncontextuality inequalities. 
\end{abstract}


\maketitle

\tableofcontents

\section{Introduction} 
\label{sec:outline}

Prepare-and-measure type experiments are essential setups corresponding to practical tasks, such as communication protocols, key distribution, computing, among many others. An important question in quantum information is finding quantum-over-classical advantages for these protocols. For this question to be rigorously approached, it is crucial to decide upon one of the various existing notions of \textit{classicality}. Throughout this work, we will consider the notion of classicality provided by generalized noncontextuality in Ref.~\cite{spekkens2005contextuality}. After it was proved that quantum theory is contextual in this sense~\cite{spekkens2005contextuality}, several works have shown that contextuality underpins advantages in quantum protocols when compared to their classical counterparts. Some examples are parity oblivious tasks~\cite{Pusey2018simplest,spekkens2009preparation}, quantum state discrimination tasks~\cite{schmid2018contextual}, state-dependent quantum cloning~\cite{lostaglio2019contextual}, linear response processes~\cite{lostaglio2020certifying} and quantum interrogation~\cite{wagner2022coherence}. Nonetheless, finding novel techniques to detect  quantum contextuality that are suitable to deal with complex scenarios is important for taking advantage of such a nonclassical feature in practical tasks.

In this work, we provide novel tools in such direction. Following the linear characterization of generalized noncontextuality from Ref.~\cite{schmid2018all} and the resource theory constructed from the underlying polytope structure presented in Ref.~\cite{duarte2018resource}, we develop techniques to reduce complex scenarios to smaller ones, where existence of contextuality serves as a witness for contextuality in the original scenario.  These tools allow us to reinterpret the results of Ref.~\cite{mazurek2016experimental}, as well as 
to witness the existence of quantum contextuality for a class of prepare-and-measure scenarios. Such a class encompasses the majority of already known proofs of contextual advantage and further generalizes it for a large class of experimental  realizations. We also identify a class of scenarios with quantum violations for \emph{all nontrivial nocontextuality inequalities}. Notably, we uncover that contextuality in the state-dependent cloning scenario of Ref.~\cite{lostaglio2019contextual} is inherited from a product of simpler scenarios. Overall, these results highlight the potential usages of our techniques for witnessing and engineering nonclassical correlations in complex scenarios lifted from simpler ones. 

In section~\ref{sec: Preliminaries}, we briefly review the notion of generalized noncontextuality. Section~\ref{Preliminaries: GC}  describes  prepare-and-measure scenarios and how noncontextual ontological models attempt at explaining the statistics arising from these experimental scenarios. Section~\ref{sec: PreliminariesMeasurementSim} describes measurement simulability in operational theories, while section~\ref{Preliminaries: RT} discusses the elements of the resource theory we are considering, with particular focus on the free operations. In section \ref{sec: Results} we describe our results in general terms, using a link between simulability and the free operations of a resource theory (section \ref{sec: ResultsSimFree}), and a binary composition of scenarios (section \ref{sec: ResultsComposition}). We then use this tools to witness and engineer quantum contextuality, in section \ref{sec: ResultsQuantumC}. In section \ref{sec:conclusions} we expose conclusions and perspectives. 

\section{Preliminaries}
\label{sec: Preliminaries}

\subsection{Generalized Contextuality}
\label{Preliminaries: GC}
Contextuality, as a notion of nonclassicality, is an inference between the operational description of an experimental setup and the ontological models one might prescribe to it. An operational theory is formally a process theory where processes are considered as lists of laboratory instructions with a probability rule~\cite{schmid2020structure}. To operationally describe a prepare-and-measure experiment, one must provide some set of laboratory procedures that prepare the studied system, following with a set of possible measurements that shall extract outcome information. The most general result from such procedures is described through conditional probabilities. We will consider the case of prepare-and-measure experiments with a finite set of  preparation procedures, that we denote by $\mathbb{P} := \{P_j\}_{j \in J}$, and a finite set of measurement procedures, $\mathbb{M} := \{M_i\}_{i \in I}$, leading to some outcome results that we label $\mathbb{O}_{\mathbb{M}} \equiv K$.  Capital letters $K,J$ and $I$  denote the set of labels with the same cardinality of the sets of primitives, $\mathbb{O}_{\mathbb{M}}, \mathbb{P},\mathbb{M}$ respectively, while $|\cdot|$ represents the cardinality of the set (e.g., $|I|$ is the number of measurement procedures in the experimental scenario). The measurement \emph{event} associated to obtaining outcome $k$ for measurement $M$ will be denoted $[k|M]$. Performing all the operations several times to acquire statistics will lead to a data-table of conditional probabilities, that we denote as $B$:
\begin{equation}
    B := \{p(k \vert M_i, P_j)\}_{k\in K, i\in I, j\in J}.
\end{equation}
We name $B$ as the \textit{behavior} of the system. Each physical realization will lead to some behavior $B$ in the set of all possible behaviors. 
Operationally, there are more structures within an experimental description. For instance, it might be so that there are operationally equivalent ways to generate some statistics; as a standard example, consider the quantum preparations $P_j$~\footnote{More precisely, each state $\rho_j$ defines an  \emph{equivalence class} $[P_j]$ of equivalent procedures implementing the same state.} associated to preparations of the quantum states $\rho^j$,
\begin{align}
     &\rho^1 = \ket{0}\bra{0}, \label{preparation 00}\\ 
    &\rho^2 = \ket{1}\bra{1}, \label{preparation 11}\\ 
    &\rho^3 = \ket{+}\bra{+},\label{preparation ++} \\ 
    &\rho^4 = \ket{-}\bra{-}.\label{preparation --} 
\end{align}

For these preparation procedures it is known that, 
\begin{equation}\label{simple operational eq}
    \frac{1}{2}\rho^1+\frac{1}{2}\rho^2 =  \frac{1}{2}\rho^3 + \frac{1}{2}\rho^4,
\end{equation}
meaning that the statistics for the measurement events $[k|M_i]$ will be the same for the above convex mixtures of $\{\rho^j\}_{j\in J}$. Such a description hints at what is understood as an \emph{operational equivalence}. 

\begin{dfn}[Operational equivalences]
Let $P,P'$ be two preparation procedures on an operational theory. Let $\mathcal{M}$ be a tomographically complete set of measurement procedures. Then, the procedures are operationally equivalent, and we write $P \simeq P'$, if and only if, 
\begin{equation}
    \forall [k|M], M\in \mathcal{M}, k\in \mathbb{O}_{\mathbb{M}}, \,\, p(k|M,P) = p(k|M,P').
\end{equation}
Equivalently, let $[k|M]$ and $[k'|M']$ be two measurement events, and $\mathcal{P}$ be a tomographically complete set of preparation procedures. Then, the events are operationally equivalent, and we write $[k|M] \simeq [k'|M']$, if and only if, 
\begin{equation}
    \forall P \in \mathcal{P}, \,\, p(k|M,P) = p(k'|M',P).
\end{equation}
\end{dfn}

Equivalences are part of the description of any operational theory, quantum or not; those denote the fact that some operational procedures (maybe defined as convex mixtures of others) cannot be distinguished using only the probabilities in the experiment.  A set of non-trivial, fixed and finitely defined operational equivalences for the preparation procedures is denoted by $\mathbb{E}_{\mathbb{P}}$, when  $a=1, \dots, \vert \mathbb{E}_{\mathbb{P}} \vert $,
\begin{equation}\label{preparation equivalences}
    \sum_j\alpha^a_jP_j \simeq \sum_j \beta^a_jP_j,
\end{equation}
where $\sum_j\alpha_j^a = \sum_j\beta_j^a = 1$, and $0\leq \alpha_j^a,\beta_j^a \leq 1$. At the level of the behaviors, we assume that convex mixtures of procedures will be respected, so that for all measurement events, we have
\begin{align}        
    \sum_{j} \alpha^{a}_{j}p(k|M_{i}P_{j}) =  \sum_{j} \beta^{a}_{j}p(k|M_{i}P_{j}).
\end{align}
Thus, each label $a$ uniquely defines a vector $\bm{\gamma}_{\mathbb{P}}^a \equiv (\bm{\alpha}^a,\bm{\beta}^a)$ associated with the preparation procedures,
\begin{equation}
    \bm{\gamma}_{\mathbb{P}}^a := (\alpha^a_1,\dots,\alpha^a_{|J|};\beta^a_1,\dots,\beta^a_{|J|}).
\end{equation}

Similarly for the measurement events, we define a set $\mathbb{E}_{\mathbb{M}}$, where the operational equivalences $b=1,\dots,\vert \mathbb{E}_{\mathbb{M}}\vert$,
\begin{equation}
    \sum_{i,k}\alpha_{[k\vert M_i]}^b[k\vert M_i] \simeq \sum_{i,k}\beta_{[k\vert M_i]}^b[k\vert M_i]
\end{equation}
uniquely define vectors $\bm{\gamma}_{\mathbb{M}}^b \equiv (\bm{\alpha}^b,\bm{\beta}^b)$, 
\begin{align}
    \bm{\gamma}_{\mathbb{M}}^b := (\alpha^b_{[1|M_1]},\dots,\alpha^b_{[|K||M_1]},\dots, \alpha^b_{[|K||M_{|I|}]}; \\
    \beta^b_{[1|M_1]},\dots,\beta^b_{[|K||M_1]},\dots, \beta^b_{[|K||M_{|I|}]}).
\end{align}

Hence, we define the sets $\mathbb{E}_{\mathbb{P}} := \left\{\bm{\gamma}_{\mathbb{P}}^{a}\right\}_{a=1}^{|\mathbb{E}_{\mathbb{P}}|}$ and $\mathbb{E}_{\mathbb{M}} := \left\{\bm{\gamma}_{\mathbb{M}}^{b}\right\}_{b=1}^{|\mathbb{E}_{\mathbb{M}}|}$, which completes the elements for the definition of a  scenario. We say that an operational equivalence $\bm\gamma$ as defined above is trivial if $\bm\alpha = \bm\beta$.

Nontriviality is required to get rid of self equivalences, since, for example, any preparation is always equivalent to itself in its own experimental setting. Hence, writing $\mathbb{E}_{\mathbb{M}} = \emptyset$ does not mean that there is no operational equivalences between measurement events, but that the experimentalist is not considering equivalences different from those of the form $M_1 \simeq M_1$ that represent no interesting constraints.

\begin{dfn}[Prepare-and-measure scenario]
A prepare-and-measure scenario is constituted by the tuple $\mathbb{B}$ given by
\begin{equation}
    \mathbb{B} = \left (\mathbb{P}, \mathbb{M}, \mathbb{O}_{\mathbb{M}}, \mathbb{E}_{\mathbb{P}}, \mathbb{E}_{\mathbb{M}}\right).
\end{equation}
\end{dfn}

Whenever it is convenient, and since we are mostly interested in the labels for the procedures, we might follow the notation of Ref.~\cite{chaturvedi2020characterising}, and write $\mathbb{B} = (|J|,|I|,|K|,\mathbb{E}_{\mathbb{P}},\mathbb{E}_{\mathbb{M}})$. Notice that scenarios do not need to have tomographically complete sets of procedures, $\mathbb{P} \subset \mathcal{P}$, but the operational equivalences must hold for $\mathcal{P}$, the complete set of procedures. 

As an example, which we shall consider when applying the techniques we develop in this work, is the simplest nontrivial scenario~\cite{Pusey2018simplest}.
\begin{dfn}[Simplest scenario, $\mathbb{B}_{\rm si}$]
\label{def: Simplest scenario}
The simplest nontrivial scenario, denoted $\mathbb{B}_{\rm si}$, is composed by $2$ dichotomic measurements $M_1$ and $M_2$ and $4$ preparation procedures, $\mathbb{P}:=\{P_i\}_{i=1}^4$. There are no equivalences for measurements, while preparations respect the equivalence relation 
 \begin{equation}
 \label{operational eq simplest scenario}
     \frac{1}{2} P_1 + \frac{1}{2}P_2 \simeq \frac{1}{2}P_3 + \frac{1}{2}P_4.
 \end{equation}
In our notation, we have $\mathbb{B}_{\rm si}:=(4,2,2,\mathbb{E}_{\mathbb{P}{\rm .si}},\emptyset)$, where $\mathbb{E}_{\mathbb{P}{\rm,si}}=\{(1/2,1/2,0,0;0,0,1/2,1/2)\}$. 
\end{dfn}

The characterization of behaviors in Ref.~\cite{schmid2018all}  provides a fundamental aspect for contextuality theory, when a finite set of operational procedures and equivalences are considered: the set of possible behaviors in $\mathbb{B}$ is in one-to-one correspondence with points in $\mathbb{R}^n$ forming a convex polytope. Inside this convex polytope of all behaviors obeying the operational equivalences, lies another polytope: the set of behaviors explained by \emph{noncontextual ontological models}.

\subsubsection{Ontological Models}
In order to \emph{explain} the conditional probabilities in a behavior $B \in \mathbb{B}$  we use the ontological models framework~\cite{harrigan2010einstein,leifer2014quantum,Kunjwal2019beyondcabello}. In such a framework, there exists some measurable set $(\Lambda,\Sigma)$ of so-called \textit{ontic states} $\lambda \in \Lambda$. These contain the full set of parameters representing the most accurate physical description of the system. From such a set of states, we construct probabilistic explanations for both preparation and measurement procedures in $\mathbb{B}$, such that:
\begin{align}
    &\forall P \in \mathbb{P}, \exists \mu_P, \\
    &\forall M \in \mathbb{M}, \forall k \in \mathbb{O}_{\mathbb{M}}, \exists \xi_{[k \vert M]},
\end{align}
where $\mu_P$ are probability measures over $(\Lambda,\Sigma)$, for all $\lambda \in \Lambda$,  $\xi_{[\cdot\vert M]}(\lambda)$ are probability distributions over the outcomes, and for any $k$, $\xi_{[k\vert M]}$ is a measurable function for $(\Lambda,\Sigma)$. Calling $\Pi$ the set of all $\mu_P$ and $\Theta$ the set of all $\xi_{[k\vert M]}$ we have that an ontological model for $B$ is a quadruple $(\Lambda, \Sigma, \Pi, \Theta)$ that recovers the conditional probabilities by means of 
\begin{equation}
    p(k\vert M_i, P_j) =  \int_\Lambda \xi_{[k\vert M_i]}(\lambda)\mathrm{d}\mu_{P_j}(\lambda) ,\forall i,j,k.
\end{equation}

The assumption of noncontextuality is defined as follows: 

\begin{dfn}[Noncontextuality]
A behavior in a prepare-and-measure scenario, $B\in \mathbb{B}$, is called  \textit{noncontextual} if there exists some ontological model $(\Sigma, \Lambda, \Pi, \Theta)$ for the behavior $B$ such that the measures from $\Pi$ respect operational equivalences of $\mathbb{E}_{\mathbb{P}}$, meaning
\begin{equation}
    \sum_j\alpha^a_jP_j \simeq \sum_j \beta^a_jP_j \Rightarrow \sum_j\alpha_j^a\mu_{P_j} = \sum_j\beta_j^a\mu_{P_j},
\end{equation}
and the same for elements from $\Theta$ that are associated with equivalent procedures from $\mathbb{E}_{\mathbb{M}}$
\begin{align}
     &\sum_{i,k}\alpha_{[k\vert M_i]}^b[k\vert M_i] \simeq \sum_{i,k}\beta_{[k\vert M_i]}^b[k\vert M_i] \Rightarrow\nonumber\\
     &\sum_{i,k}\alpha_{[k\vert M_i]}^b\xi_{[k\vert M_i]} = \sum_{i,k}\beta_{[k\vert M_i]}^b\xi_{[k\vert M_i]}.
\end{align}
\end{dfn}

The set of behaviors that do have a noncontextual ontological explanation is fully characterized by a finite set of tight inequalities forming the so-called noncontextual polytope $NC(\mathbb{B})$~\cite{schmid2018all}, see   Appendix \ref{sec:appendix2} for an example. The behaviors that are incompatible with any noncontextual ontological explanation are said to be contextual. It is already established that operational descriptions of quantum theory, where POVMs represent measurements and density matrices represent preparations, can lead to contextual behaviors (see Refs.~\cite{leifer2014quantum,Banik_2014,mazurek2016experimental})-- in particular, the simplest scenario in which such a violation can occur is given by $\mathbb{B}_{\rm si}$. 

Deciding if a behavior in a given scenario is noncontextual or not can be framed as a linear program, and is fully determined by the complete set of facet-defining noncontextuality inequalities characterizing the polytope $NC(\mathbb{B})$~\cite{schmid2018all}. Using hierarchies of semi-definite programs (SDP), it is also possible to bound the set of quantum behaviors~\cite{chaturvedi2020characterising,tavakoli2021bounding}. However, in most of the situations, these numerical tools provide little intuition for generating novel analytical insights, and they become increasingly computationally demanding. This is especially evident when attempting at finding all noncontextuality inequalities or applying SDP hierarchies to  scenarios where $|I|, |J|, |K|\gg 2$.

The resource theoretic toolbox will be instrumental to provide simple yet important qualitative understanding of possibly large prepare-and-measure scenarios, while avoiding numerically demanding procedures. We will do so by leveraging the concept of measurement simulability, and by lifting inequalities present in smaller scenarios into more complex ones.  Before delving into the key aspects of the resource theory of generalized contextuality, we will first explain measurement simulation within an operational-probabilistic theories perspective.

\subsection{Measurement simulability}
\label{sec: PreliminariesMeasurementSim}

One notion that will be valuable to this work is that of measurement simulability. It was first stated for quantum measurements~\cite{Guerini2017,Oszmaniec_2017} and recently studied in the context of generalized probabilistic theories~\cite{Filippov_SimulabilityGPT2018}. The basic idea is to understand which measurement statistics can be obtained by using a given set of measurement apparatuses and classical (pre- or post-) processing. Here we adapt the notion of measurement simulability to operational theories.

\begin{dfn}[Measurement simulability]
\label{def: MeasurementSimulability}
Consider a set of $|I|$ measurement procedures $\mathbb{N}\equiv \{N_i\}_{i\in I}$, on a given operational theory, with outcome set $K$. Then, another measurement procedure set $\{M_{\tilde{i}}\}_{\tilde{i} \in \tilde{I}}$ on this operational theory, with outcome set $\tilde{K}$, is said to be $\mathbb{N}$-simulable if there exists classical pre-processings $q_M(i|\tilde{i})$ and post-processing $q_O^i(\tilde{k}|k)$  such that 
\begin{equation}
    [\tilde{k}|M_{\tilde{i}}]\simeq\sum_{i,k}  q_O^i(\tilde{k}|k)[k|N_i]q_M(i|\tilde{i})
\end{equation}
for every $\tilde{k}\in \tilde{K}$ and $\tilde{i}\in\tilde{I}$. Above, $q_M(i|\tilde{i})$ is a conditional probability that, for each $\tilde{i}$, chooses $N_i$ with probability $q_M(i|\tilde{i})$.  Similarly, $q_O^i(\tilde{k}|k)$ define the probability of outcome $\tilde{k}\in \tilde{K}$ given $k\in K$ for each $i\in I$.
\end{dfn}

In other words, measurement simulability states \emph{an  equivalence of specific form}  between the  measurement procedure to be simulated and the set of measurements performing the simulation, in which the coefficients defining the equivalence are decomposed as $\beta^{\,{\rm sim}\tilde{M}_{\tilde{i}}}_{[k_i|N_i]}:=q_M(i|\tilde{i})q_O^i(\tilde{k}|k)$.

As said before, the definition of simulability has been studied in the framework of generalized probabilistic theories~\cite{Filippov_SimulabilityGPT2018}, which is similar to the above definition, but with measurement events replaced by effects. If one has access to a tomographically complete set of measurements, our definition naturally implies the definition for the GPT associated to the operational theory by quotienting contexts~\cite{schmid2020characterization}.
In the case of quantum theory (in tomographically complete scenarios), measurement $\{N_i\}_{i \in I}$-simulability of a measurement set $\{M\}_{\tilde{i} \in \tilde{I}}$ can be obtained by interchanging $[\tilde{k}|M_{\tilde{i}}]$ and $[k|N_i]$ with the corresponding POVM elements.

\subsection{Resource theory}
\label{Preliminaries: RT}

In general formulations of resource theories, the basic ingredients are  \textit{objects}, that may feature a specific resource, as well as operations among those objects~\cite{duarte2018resource, coecke2016mathematical}. Objects without any resource, and operations incapable to create them are called, respectively, free objects and free operations. Free operations define a pre-order: if an object $o$ can be freely transformed into $o'$, then $o$ must have at least the same amount of resources as $o'$. This pre-order, in turn, must be respected by any monotone aiming to quantify the resource.

For the present work, we consider contextuality in any fixed prepare-and-measure scenario $\mathbb{B}$ as the resource, following Ref.~\cite{duarte2018resource}. Thus, we are interested in considering the objects as $B \in \mathbb{B}$, while the set of free objects is naturally defined by the polytope $NC(\mathbb{B})$. 
The set of free operations defining the resource theory we consider is the set of pre-processing preparations or measurements, together with a post-processing of the measurement results~\cite{duarte2018resource}.

\begin{dfn}
\label{def: free operations}
Given a scenario $\mathbb{B}:= \left(|J|,|I|, |K|, \mathbb{E}_{\mathbb{P}}, \mathbb{E}_{\mathbb{M}}\right)$ we define the set of free operations $\mathcal{F}$ as the set of maps $T:\mathbb{B} \to T(\mathbb{B})$ such that
\begin{align}
    &\hspace{0.8cm}T:\left \{ p(k \vert M_i, P_j)\right\}_{k\in K, i\in I, j\in J } \mapsto\nonumber \\
    &\left\{ \sum_{i,j,k}q_O^i(\tilde{k}\vert k) p(k\vert M_i,P_j)q_M(i\vert \tilde{i})q_P(j\vert \tilde{j})\right\}_{\tilde{k}\in \tilde{K}, \tilde{i}\in \tilde{I}, \tilde{j}\in \tilde{J}}
    \label{eq: FreeOperations}
\end{align}
where $q_O^i: K \to \tilde{K}, q_M: \tilde{I}\to I, q_P: \tilde{J}\to J$ are stochastic maps between index sets, i.e. $q_P=(q_P(j\vert \tilde{j}))_{j,\tilde{j}}$ is a stochastic matrix, corresponding to operational primitives in the different scenarios defined for each $\mathbb{B}$ by $T(\mathbb{B}):=\left(|\tilde{J}|, |\tilde{I}|, |\tilde{K}|, \mathbb{E}_{T(\mathbb{P})}, \mathbb{E}_{T(\mathbb{M})}\right)$, for set's of operational equivalences defined for the procedures \textit{after} the transformation $T$ was performed. 
\end{dfn}

The free operations have an impact on the equivalence classes. For instance, the new coefficients for preparations, $\tilde{\bm{\alpha}}$ and $\tilde{\bm{\beta}}$ (for every $s$ labeling the equivalences),  are those obeying equations~\cite{duarte2018resource} 
\begin{subequations}
\label{eq: UpdateEquivalences}
\begin{equation}
    \label{eq: updateAlpha}
    \alpha_j^s = \sum_{\tilde{j}\in \tilde{J}}\tilde{\alpha}_{\tilde{j}}^s q_P(j\vert \tilde{j}),
    \end{equation}
    \begin{equation}
    \label{eq: updateBeta}
    \beta_j^s = \sum_{\tilde{j}\in \tilde{J}}\tilde{\beta}_{\tilde{j}}^s q_P(j\vert \tilde{j}),
    \end{equation}
\end{subequations}
where $q_P(j|\tilde{j})$ are defined by the free operation (with similar relations for equivalences on measurements). The change in the operational equivalences is represented by the notation $\mathbb{E}_{\mathbb{P}} \stackrel{T}{\to} \mathbb{E}_{T(\mathbb{P})}$. Some particularly important features of the new equivalences are: first, non-trivial equivalences in the new scenario can appear, even if the original scenario had trivial equivalences. This is so because  $\bm{\alpha} = \bm{\beta}$ \textit{does not} imply that $q_P\bm{\alpha}= \tilde{\bm{\alpha}} = \tilde{\bm{\beta}} = q_P\bm{\beta}$, for $q_P$ a (left) stochastic matrix. Second, no equivalences can be `broken'; indeed, from \eqref{preparation equivalences},
\begin{align*}\label{uplifted operations for preparations}
    &\sum_{j \in J} \alpha_j^sP_j\simeq \beta_j^s P_j \\
    &\implies \sum_{j,\tilde{j}}\tilde{\alpha}^s_{\tilde{j}}q_P(j\vert \tilde{j})P_j\simeq \sum_{j,\tilde{j}}\tilde{\beta}^s_{\tilde{j}}q_P(j\vert \tilde{j})P_j \\
    &\implies \sum_{{\tilde{j}} \in \tilde{J}}\tilde{\alpha}^s_{\tilde{j}}\tilde{P}_{\tilde{j}}\simeq \sum_{{\tilde{j}} \in \tilde{J}}\tilde{\beta}^s_{\tilde{j}}\tilde{P}_{\tilde{j}} 
\end{align*}
where the new set of preparations $T(\mathbb{P}) \equiv \{P_{\tilde{j}}\}_{\tilde{j} \in \tilde{J}}$ are defined in the new scenario $T(\mathbb{B})$. 

Finally, there are some different examples of monotones respecting the pre-order established by the free operations. The one we will use on this work is the $l_1-$distance from Ref.~\cite{duarte2018resource}.
\begin{dfn}
\label{def: l1-distance}
Let $\mathbb{B} := (|J|,|I|,|K|,\mathbb{E}_{\mathbb{P}},\mathbb{E}_{\mathbb{M}})$ be any finitely defined prepare-and-measure scenario. The $l_1$-contextuality distance $\mathsf{d}: \mathbb{B} \to \mathbb{R}_+$ is defined by
\begin{equation}
     \mathsf{d}(B) := \min_{B^* \in NC(\mathbb{B})}\max_{i \in I,j \in J}\sum_{k\in K} \vert p(k|M_i,P_j) - p^*(k|M_i,P_j) \vert 
\end{equation}
\end{dfn}

In Ref.~\cite{baldi2021emergence} this measure was used to bound nonclassicality in finite scenarios relevant for quantum Darwinism.

\section{Results}
\label{sec: Results}

The results here reported are essentially obtained by exploring the defining feature of free operations; namely, that they cannot increase the resource (contextuality). The practical implication is that if $T$ is a free operation and $T(B)\in T(\mathbb{B})$ is contextual, then $B$ must be contextual on the original scenario, $\mathbb{B}$.

With this in mind, we first show how to reduce some complex scenarios to simpler ones by using measurement simulability. A practical implication of this is that we can attain/explore contextuality with easier implementations; with this perspective we reinterpret the results of Ref.~\cite{mazurek2016experimental}. Secondly, we take the opposite path, showing how to build more complex scenarios from simpler ones -- where important features of the simple scenarios are carried to the complex one. This technique allows to engineer scenarios where all non-trivial facets exhibit quantum violations. Moreover, we conclude that the contextual advantage on the cloning task~\cite{lostaglio2019contextual} is inherited from a simpler scenario. 

\subsection{Simulability and free operations}
\label{sec: ResultsSimFree}

We begin by showing that measurement simulation, as expressed in Def.~\ref{def: MeasurementSimulability}, physically implements a  subset of the free operations.

\begin{lemma}[Simulation is free]
\label{lemma: SimFree}
Consider a $\{N_i\}_{i \in I}$-simulation of a set $\{M_{\tilde{i}}\}_{\tilde{i}\in \tilde{I}}$ and a set of preparations $\mathbb{P}$. Now consider the behaviors obtained by the simulating set $\{N_i\}_{i \in I}$, $B_{N}:=\{p(k|N_{i},P_{j})\}_{k\in I,i\in I,j\in J}$, and those obtained by the simulated set $\{M_{\tilde{i}}\}_{\tilde{i}\in \tilde{I}} $, $B_{M}:=\{p(\tilde{k}|M_{\tilde{i}},P_j)\}_{{\tilde{k}\in\tilde{K}},{\tilde{i}\in \tilde{I}},j\in J}$. 
The operation implemented by such a simulation, $T_{\rm sim}: B_{N}\mapsto B_{M}$, is free.
\end{lemma}
\begin{proof} Measurement simulation acts as a map $T_{\rm sim}$ which takes the measurement events $[k|N_i]$, to the measurement events $\sum_{k,i}q_O^i(\tilde{k}|k)[k|N_i]q_M(i|\tilde{i})$. Due to linearity, the impact of simulation on behaviors is $T_{\rm sim}(\{p(k|N_i,P_j)\})= \{\sum_{k,i}q_O^i(\tilde{k}|k)p(k|N_i,P_j)q_M(i|\tilde{i})\}$. Now, the equivalence established by simulation, $[\tilde{k}|M_{\tilde{i}}]\simeq \sum_{k,i}q_O^i(\tilde{k}|k)[k|N_i]q_M(i|\tilde{i})$, implies
\begin{equation}
    \sum_{k,i} q_O^i(\tilde{k}|k)p(k|N_i,P_j)q_M(i|\tilde{i})=p(\tilde{k}|M_{\tilde{i}},P_j)\,\,\forall P_j \in \mathbb{P}.
     \label{eq: SimOnbehaviors}
\end{equation}
By comparing the l.h.s. of Eq.~\eqref{eq: SimOnbehaviors} with the r.h.s. of Eq. \eqref{eq: FreeOperations}, we see that $T_{\rm sim}$ is indeed a specific kind of free operation (obtained through simulation), which leaves preparations untouched.
\end{proof}

This lemma has a direct implication for quantum realizations (we will denote $\mathbb{M}^Q$ as the quantum realizations of the procedures $\mathbb{M}$):
\begin{coro}
Let $\mathbb{M}_1^Q,\mathbb{M}_2^Q$ be sets of quantum realizations of  prepare-and-measure scenarios $\mathbb{B}_1,\mathbb{B}_2$, respectively. Then, if $\mathbb{M}_1^Q$ is $\mathbb{M}_2^Q$-simulable there exists a free operation $T:\mathbb{B}_2 \to T(\mathbb{B}_2)=\mathbb{B}_1$. 
\end{coro}

One implication of the above results is that
one can use simulations to derive simpler scenarios from complex ones. We discuss two simple instances of how this can be done: by manipulating measurement equivalences or by moving to a scenario having fewer measurements than the original one.

The first instance is a consequence of the impact of free operations on equivalence classes (as expressed in Eqs.~\eqref{eq: UpdateEquivalences} for preparations). Indeed, by performing classical pre- and post-selection of events one might be able to engineer the equivalence classes of interest. This gives an alternative interpretation of the results of Ref.~\cite{mazurek2016experimental}, that we proceed to briefly recall: In their work, the authors tackle the problem of practical impossibility to obey exactly the desired operational equivalences for the ideal quantum procedures, due to experimental errors. Assume that we want to test a noncontextuality inequality defined for a scenario $\mathbb{B}$. When operationally characterizing the procedures of $\mathbb{B}$ in a real experiment the noisy data effectively implements some other closely related scenario $\mathbb{B}^{\rm p}$. We call these the \textit{primary procedures}. For concreteness, let us use preparations $\mathbb{P}^{\rm p}$ to express in precise terms the idea. The procedures $\mathbb{P}^{\rm p}$ correspond to those that can be characterized using the (robust) experimental implementations. In particular, the problem of those is that they do not satisfy the ideal operational equivalences of the target scenario $\mathbb{B}$, with preparation procedures $\mathbb{P}$, in which case the noncontextuality inequality tested is not applicable.

By performing classical post-processing in the procedures, it is possible to obtain new \textit{secondary} procedures which match the expected operational equivalences perfectly, \emph{by construction}. The mapping can be framed as something of the form 
\begin{equation}\label{eq: secondary to primary Ps}
    P^{\rm s}_{\tilde{j}} = \sum_j q_P(\tilde{j}|j) P_j^{\rm p}
\end{equation}
for all $j \in J$ labeling the elements of $\mathbb{P}^{\rm p}$. Properly choosing $p(\tilde{j}|j)$ allows the procedures $\mathbb{P}^{\rm s} := \{P^{\rm s}_{\tilde j}\}_{\tilde j}$ to satisfy the target operational equivalences of $\mathbb{B}$. With this, the behavior obtained from the secondary procedures can now be properly used to violate the inequality, that is now applicable.

To this approach one could provide the following criticism: Since we never obtain a noncontextual bound with respect to the primary (measured) procedures and its corresponding operational equivalences, what guarantees that we are not demonstrating contextuality of the secondary procedures only? The resource theory framework guarantees that:
\begin{theorem}
Contextuality of the behaviors obtained with the secondary procedures imply contextuality of the behaviors obtained from the primary procedures.
\end{theorem}
\begin{proof}
Recall that for the monotone $\mathsf{d}$ it is true that $\mathsf{d}(T(B)) \leq \mathsf{d}(B)$ for all $B \in \mathbb{B}$ and $T\in \mathcal{F}$ free operation. Transformations from primary to secondary procedures are of the form given by Def.~\ref{def: free operations}. Let us denote this operations as $T_{\text{p} \to \text{s}}$. This can be seen simply by noticing that $T_{\text{p} \to \text{s}}$   probabilistic mixes the secondary procedures given the primary ones, as is expressed by Eq.~\eqref{eq: secondary to primary Ps}. Since $\forall T \in \mathcal{F}$, it is true that $\mathsf{d}(T(B)) > 0 \implies \mathsf{d}(B) > 0$,  the fact that $B_{\text{s}} = T_{\text{p} \to \text{s}} (B_{\text{p}})$ implies $\mathsf{d}(B_{\text{s}}) > 0 \implies \mathsf{d}(B_{\text{p}}) > 0 \implies B_{\text{re}}$ is contextual. 
\end{proof}
With the resource-theoretic perspective here proposed, we can understand the methods of Ref.~\cite{mazurek2016experimental} as using a free operation to obtain new behaviors which obey the desired operational equivalences and still exhibit contextuality. Moreover, since the performed operation is free, we can add that their violations \emph{also show contextuality for the original measurements, in the original scenario.} Notice that imposing assumptions on the possible experimental errors this argument can be extended to the ideal quantum realizations. 

Let us now consider the use of measurement simulation to reduce a given scenario. We will consider the simplest case, where part of the measurements are erased. The following results are corollaries of Lemma~\ref{lemma: SimFree}.
\begin{coro}[Trivial simulation] \label{corollary: trivialSim}
Consider a scenario $\mathbb{B}=(|J|,|I|,|K|,\mathbb{E_P},\mathbb{E_M})$ and define another scenario $\mathbb{B'}$ obtained simply by discarding some of the measurements, i.e. $\mathbb{B'}=(|J|,|I'|,|K|,\mathbb{E_P},\mathbb{E'_{M'}})$  where $\mathbb{M}'\subset\mathbb{M}$ (thus, $|I'|\leq |I|)$ and $\mathbb{E'_{M'}}\subset\mathbb{E_M}$. The transformation of erasing such procedures and equivalences among them, $T: \mathbb{B}\mapsto \mathbb{B'}$, is free.
\begin{proof}
This transformation can be mathematically described as $T(B)=\{\sum_{j} p(k|N_i,P_j)q_M(i|i')\}$ where $q(i|i')=1$ if $i=i'$ and $0$ otherwise. 
\end{proof}
\end{coro}
This will be of importance to us, specially in the case where all remaining measurements are dichotomic and with no equivalences, $\mathbb{E_M'}=\emptyset$. In other words, the case where one arises at simple generalizations of the simplest scenario \ref{def: Simplest scenario} after discarding a subset of measurements. In this case, we know that there is a quantum realization of the measurements of the reduced scenario. That is, 

\begin{coro}
\label{corollary: trivialSimQ}
Let $\mathbb{M}:= \{M_i\}$ be any set of two-outcome operational measurements having a quantum realization $\mathbb{M}^Q$. Then, the quantum measurements  $\mathbb{M}^Q_{\rm si} := \left \{\frac{1}{\sqrt{2}}(\sigma_X + \sigma_Z), \frac{1}{\sqrt{2}}(\sigma_X - \sigma_Z)\right \} $ are $\mathbb{M}^Q$-simulable, for at least some quantum realization $\mathbb{M}^Q$ of $\mathbb{M}$ (those of the form $\mathbb{M}^Q=\mathbb{M}^Q_{\rm si}\bigcup\mathbb{N}^Q$ for some other set of quantum realizations $\mathbb{N}^Q$). 
\end{coro}

With the above results, we see that using simulations provided by particular set of measurements may lead us to new, simpler, scenarios. If contextuality is witnessed in such scenarios, the resource theoretical perspective allows to conclude that contextuality was present prior to the simplification process. We will use such ideas in section~\ref{sec: ResultsQuantumC} to engineer and witness quantum contextuality on more involved scenarios.

\subsection{Composition of scenarios}
\label{sec: ResultsComposition}

In the previous section, we discussed how to use free operations (measurement simulability in particular) to obtain simpler scenarios. Here we take the opposite path, constructing complex scenarios from simpler ones. This particular construction allows to obtain important information regarding the resource, which is inherited from the original, smaller, scenarios. This is based on the following definition,
{
\begin{dfn}
\label{def: box scenario}
Let $\mathbb{B}_1=(|J_1|,|I_1|,|K_1|,\mathbb{E}_{\mathbb{P}_1},\mathbb{E}_{\mathbb{M}_1})$, and $\mathbb{B}_2=(|J_2|,|I_2|,|K_2|,\mathbb{E}_{\mathbb{P}_2},\mathbb{E}_{\mathbb{M}_2})$ be two finitely defined prepare-and-measure scenarios, with behaviors $B_1\in \mathbb{B}_1,B_2 \in \mathbb{B}_2$, seen as vectors $B_1 = (p(k_1|M_{i_1},P_{j_1}))_{i_1,k_1,j_1}$, we then define:
\begin{enumerate}
\item[(a)] (Composition of scenarios) The target scenario $\mathbb{B} \equiv \mathbb{B}_1 \boxplus \mathbb{B}_2$ defined by the tuple, $$(\vert J_1\cup J_2\vert,\vert I_1\cup I_2\vert,\vert K_1\cup K_2\vert,\mathbb{E}_{\mathbb{P}_1\cup\mathbb{P}_2},\mathbb{E}_{\mathbb{M}_1\cup\mathbb{M}_2})$$ has the operational equivalences of both scenarios defined as, for $\{a\}_{a=1}^{|\mathbb{E}_{\mathbb{P}_1 \cup \mathbb{P}_2}|} := \{a_1\}_{a_1=1}^{|\mathbb{E}_{\mathbb{P}_1}|} \cup \{a_2\}_{a_2=1}^{|\mathbb{E}_{\mathbb{P}_2}|}$, that we denote $\bm{\gamma}_{\mathbb{P}_1\cup\mathbb{P}_2}^a \in \mathbb{E}_{\mathbb{P}_1\cup\mathbb{P}_2}$,
\begin{equation}\label{op eq box scenario}
    \bm{\gamma}_{\mathbb{P}_1\cup\mathbb{P}_2}^a := \left\{\begin{matrix}
    (\bm{\alpha}^{a_1},\bm{0};\bm{\beta}^{a_1},\bm{0}), & a = a_1 \\
    (\bm{0},\bm{\alpha}^{a_2},\bm{0};\bm{\beta}^{a_2}), & a= a_2
    \end{matrix}\right.
\end{equation}
The  analogous definition holds for the operational equivalences for measurement events by a change $P \to M$ and $a \to b$.
\item[(b)] (Composition of behaviors) The binary operation $\boxplus$ is defined as the vertical stacking of vectors from the scenarios $\mathbb{B}_1,\mathbb{B}_2$ towards $\mathbb{B}$, i.e., 
\begin{equation}\label{def: box product}
    B_1 \boxplus B_2 := \left(
    \begin{matrix} p(k_1|M_{i_1},P_{j_1}) \\ p(k_2|M_{i_2},P_{j_2})
    \end{matrix}\right)
\end{equation}
With $i_1 \in I_1$, $\vert I_1\vert = \vert \mathbb{M}_1\vert$ and similarly for all other labels.
\end{enumerate}

\end{dfn}
}
As an operational constraint, the target scenario \textit{does not}  consider the probabilities obtained with hybrid procedures, i.e. those of the form 
 \begin{equation}
     p(k_1|M_{i_1},P_{j_2}),p(k_2|M_{i_1},P_{j_1}) \notin B_1 \boxplus B_2.
 \end{equation}
Note that compositions of multiple scenarios is constructed in sequence and is associative. 

This binary operation essentially appends two given scenarios. The geometrical consequences of such composition will be important, as it will allow us to build an intuition of the resulting noncontextual polytope. This is described by the following lemma:

\begin{lemma}[Geometrical consequences, from Refs.~\cite{brondsted2012introduction,paffenholz2006new,eu2020mestrado}]
\label{lemma: product convex polytopes}
Let $P \subset \mathbb{R}^n$, $Q \subset \mathbb{R}^m$ be two convex polytopes. Then, the product defined by 
\begin{equation}
    P \times Q := \left\{ \left(\begin{matrix}
    p \\ q
    \end{matrix}\right): p \in P, q \in Q\right\} \subset \mathbb{R}^{n+m},
\end{equation}
is again a convex polytope. Let $\vert V(P) \vert $ and $\vert V(Q) \vert$ represent the number of vertices of each of the convex polytopes $P$ and $Q$, then, we also have that $\vert V(P\times Q)\vert = \vert V(P) \vert \cdot \vert V(Q) \vert$. Let $\vert F(P)\vert $ define the number of facets of the convex polytope $P$, and similarly for the convex polytope $Q$. Then, we have that $\vert F(P \times Q) \vert = \vert F(P)\vert  + \vert F(Q)\vert $.  
\end{lemma}

One important feature of the binary operation is the following result.
\begin{theorem}
\label{theo:block preserve resource}
The binary operation $\boxplus$ preserves the resource: 
\begin{equation*}
    B_1\in NC(\mathbb{B}_1),B_2\in NC(\mathbb{B}_2) \Leftrightarrow B_1\boxplus B_2 \in NC(\mathbb{B}_1 \boxplus \mathbb{B}_2).
\end{equation*}
\end{theorem}

We prove this theorem in Appendix~\ref{appenidx: blocks}, but we provide some intuition here. $(\Rightarrow)$ If  there exists a noncontextual ontological model for each part $B_1$ and $B_2$, respecting the operational equivalences of both scenarios, then in the new scenario $\mathbb{B}:=\mathbb{B}_1\boxplus\mathbb{B}_2$ the operational equivalences are inherited according to \eqref{op eq box scenario}, so that we can choose the set of ontic states $\Lambda_1 \sqcup \Lambda_2$, and construct a noncontextual model for any $B$ using the model of the parts, over this larger ontic space. $(\Leftarrow)$ Now, on the other way around, if there exists a noncontextual ontological model for any behavior $B_1 \boxplus B_2$ there must exist one for its parts, by simply restricting the probability distributions to the correct labels since the operational equivalences are of the form of equation \eqref{op eq box scenario}.
An implication of $\boxplus$ preserving the resource is that it does not increase the $l_1$-distance quantifier:
\begin{lemma}
\label{lemma:Compositionl1-distance}
Let $\mathbb{B} := (|J|,|I|,|K|,\mathbb{E}_{\mathbb{P}},\mathbb{E}_{\mathbb{M}})$ be any finitely defined prepare-and-measure scenario. The $l_1$-contextuality distance $\mathsf{d}: \mathbb{B} \to \mathbb{R}_+$ (see Def.~\ref{def: l1-distance})
is subbaditive under the binary operation $\boxplus$. This means that for $B_1 \in \mathbb{B}_1$ and $B_2 \in \mathbb{B}_2$ we get:
\begin{equation}
     \mathsf{d}(B_1 \boxplus B_2) \leq \mathsf{d}(B_1)+ \mathsf{d}(B_2).
\end{equation}
\end{lemma}

In other words, Theorem~\ref{theo:block preserve resource} (and its manifestation through the $l_1$-distance) tells us that this composition preserves the structure of the noncontextual behaviors; therefore, if $B_1\boxplus B_2\in \mathbb{B}_1\boxplus \mathbb{B}_2$ is contextual, then it must be true that either $B_1$ or $B_2$ is contextual (or both). 

Lemma~\ref{lemma: product convex polytopes}  and Theorem~\ref{theo:block preserve resource} give interesting tools to understand a complex scenario. Indeed, if we are able to decompose a given scenario as $\mathbb{B}=\mathbb{B}_1\boxplus\mathbb{B}_2\boxplus\ldots\boxplus \mathbb{B}_n$, we can obtain resourceful behaviors on $\mathbb{B}$ by building on resourceful behaviors on its components, $\mathbb{B}_l$($l\in\{1,\ldots,n\}$).  
%

%
\subsection{Witnessing quantum contextuality}
\label{sec: ResultsQuantumC}
In sections~\ref{sec: ResultsSimFree} and \ref{sec: ResultsComposition}, we exposed general results which show how a resource-theoretic approach provides interesting tools to analyze complex contextuality scenarios. Namely, by reducing a scenario via erasing procedures, designing equivalences or looking for a nice decomposition in terms of the product \eqref{def: box product}; we can also use such a product to build up complex scenarios which preserve the resource. Here we take advantage of those results to engineer and witness \emph{quantum} contextuality. In particular, we show that scenarios of a particular form always feature quantum contextuality and that the contextual advantage present in the cloning scenario~\cite{lostaglio2019contextual} is actually inherited from a simpler scenario (that we name $\mathbb{B}_6$).

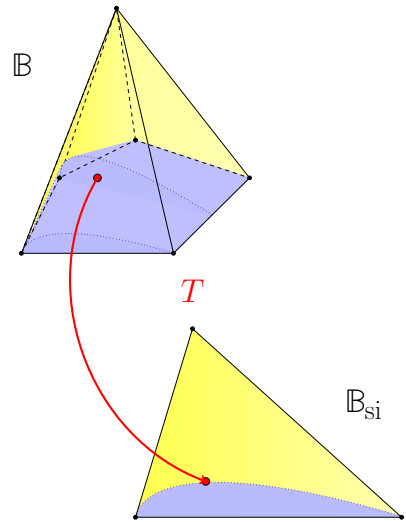
\begin{figure}[H]
\centering
\scalebox{0.5}{ 
\begin{tikzpicture}
\shade[right color=yellow!40,left color=yellow!70] (-5,3)--(-3.5,7.5)--(-3,4);
\draw[fill=yellow!70,dashed] (-5,3)--(-3.5,7.5)--(-6,1);
\shade[right color=yellow!30,left color=yellow!40] (-3,4)--(-3.5,7.5)--(0,3);
\draw[dotted,blue,fill=blue!25] (-6,1)--(-5,3)--(-0.98,2.03)--(-2,1);
\draw[blue!25,fill=blue!25] (-5,3)--(-4.65,3.55)--(-3,4)--(0,3)--(-0.98,2.03)--(-5,3);
\draw[dotted,blue,fill=blue!27,samples=50,domain=-6:-2] plot (\x,{1.518*((\x+6)^(0.7))-\x-6+1});
\draw[dotted,blue,fill=blue!27,samples=75,domain=-5:-0.98] plot (\x,{1.518*((\x+5)^(0.5))-\x-6+4});
\draw[dotted,blue,fill=blue!20,domain=-6:-4.93] plot (\x,{2.2*\x+14.2})--(-5,3)--(-6,1);
\draw[fill] (-2,1) circle [radius=0.05]; 
\draw[fill] (-6,1) circle [radius=0.05]; 
\draw[fill] (0,3) circle [radius=0.05]; 
\draw[fill] (-5,3) circle [radius=0.05]; 
\draw[fill] (-3,4) circle [radius=0.05]; 
\draw[fill] (-3.5,7.5) circle [radius=0.05]; 
\draw (-6,1) -- (-2,1)--(0,3); 
\draw[dashed] (-6,1)--(-5,3)--(-3,4)--(0,3);
\draw (-2,1)--(-3.5,7.5);
\draw (-6,1)--(-3.5,7.5);
\draw (0,3)--(-3.5,7.5);
\draw[dashed] (-3,4)--(-3.5,7.5);

\shade[right color=yellow!30,left color=yellow!70](-1.5,-1)--(-3,-6)--(4,-6);
\draw[dotted,blue,fill=blue!27,samples=75,domain=-3:4.1] plot (\x,{1.8*((\x+3)^(0.7))-\x-9});
\draw(-1.5,-1)--(-3,-6)--(4,-6)--(-1.5,-1);
\draw[fill] (-1.5,-1) circle [radius=0.05]; 
\draw[fill] (-3,-6) circle [radius=0.05]; 
\draw[fill] (4,-6) circle [radius=0.05]; 

\draw[fill=red] (-1.15,-5.05) circle [radius=0.1]; 
\draw[fill=red] (-4,3) circle [radius=0.1]; 

\node at (-6,6) {\Huge $\mathbb{B}$};
\node at (3,-3) {\Huge $\mathbb{B}_{\rm si}$};
\node[red] at (-1.5,0) {\Huge $T$};
\coordinate   (A) at (-4,3) ;
\coordinate  (B) at (-1.15,-5.05) ;
\draw[ultra thick,red,->] (A) to [bend right=50]  (B) ;
\end{tikzpicture}
}
\caption{\label{figure: Free transformation Pic}\textbf{Free operations as a tool for witnessing quantum contextuality in  complex scenarios.} By finding the existence of a free transformation $T$ towards quantum contextual behaviors in already known scenarios, such as the simplest scenario $\mathbb{B}_{\rm si}$ from Ref.~\cite{Pusey2018simplest}, one can attest contextuality in the original case.}
\end{figure}

\subsubsection{Using free operations}

Here we show examples in which one can prove existence of quantum contextuality in certain scenarios by taking advantage of free operations in the resource-theoretic approach. In particular, we use free operations (such as the trivial simulation discussed in Corollary~\ref{corollary: trivialSim}), to take these scenarios to the simplest one ($\mathbb{B}_{\rm si}$, see Def.~\ref{def: Simplest scenario}), and still find contextual quantum realizations. The idea is represented in figure \ref{figure: Free transformation Pic}.

In what follows, we discuss scenarios of a specific structure in which we can apply such technique.

\begin{exm}\label{example: BtowardsSi_M}
Consider scenarios of the kind $\mathbb{B}:= (4,|I|,2,\mathbb{E}_{\mathbb{P},si},\emptyset), |I| \geq 2$. There is always a quantum behavior $B \in \mathbb{B}$ and a free operation towards the simplest scenario, i.e. $T\in \mathcal{F}$  where $T(\mathbb{B}) = \mathbb{B}_{\rm si}$, such that  $T(B) \in \mathbb{B}_{\rm si}$ is a quantum  contextual behavior. 
\end{exm}
In other words, every prepare-and-measure scenario that can be written as above  will have at least some quantum realization which is  contextual.  Such contextual behavior \textit{may be} understood as a quantum advantage in such a scenario, as suggested by the resource-theoretic approach.

\begin{proof}

Since there are no equivalences among measurements, we can choose any set $\mathbb{M}^Q$ with $|I|$ procedures to be a realization of the measurement procedures $\mathbb{M}$. Finally, we can use the strategy of trivial simulation described in Corollary~\ref{corollary: trivialSim} to take $\mathbb{B}$ to $\mathbb{B}_{\rm si}$, simply by erasing all but two measurements. Finally, we can use the quantum realizations providing contextual advantage in this scenario, which proves our example.

Even though the above completes our proof, we might profit from an explicit description of such procedure.
Let a quantum realization of the measurement procedures $\mathbb{M}=\{M_i\}_{i\in I}$ from the scenario $\mathbb{B}$ be such that each measurement is a projective measurement, with $\mathbb{M}^Q = \mathbb{M}^Q_{\rm si} \cup \{M^{proj}_{3},\dots,M^{proj}_I\}$, for $M^{proj}_{i}$ a projective measurement for all $i=3,\dots,\vert I\vert $. Now, define maps $q_O^i, q_M$ as
\begin{align}
\label{eq: trivialSimQ}
    &q^i_O(\tilde{k}|k)=\delta_{k,\tilde{k}}\\
    &q_M(i|\tilde{i})=\begin{cases}1 \text{ if } i\in\{0,1,2\}\\
    0, \text{otherwise}.
    \end{cases}
\end{align}
 For any measurement event represented by the quantum operators $E_{\tilde{k}}^{\tilde{i}}$ from the POVMs $M_{\tilde{i}} \in \mathbb{M}^Q_{\rm si}$, we have
\begin{equation}\label{operational equation theorem}
     E_{\tilde{k}}^{\tilde{i}} = \sum_{i,k}q_O^i(\tilde{k}\vert k)E^i_k q_M(i\vert \tilde{i}). 
\end{equation}
Let the quantum realization of the preparation procedures in $\mathbb{B}$ be that given in equations \eqref{preparation 00}-\eqref{preparation --}. These can also be a quantum realization for $\mathbb{B}_{\rm si}$, since the preparation structure of both scenarios is the same. Then, if we define the free operation $T$ by means of the maps in equations \eqref{operational equation theorem} and \eqref{eq: trivialSimQ}, we can notice that the following holds,
\begin{align*}
    p(k\vert M_i, P_j) &= \text{Tr}\left ( E_k^i \rho^j\right) \stackrel{T}{\to} \\
    &\to\sum_{i,k}q_O^i(\tilde{k}\vert k)\text{Tr}\left(E_k^i\rho^j\right)q_M(i\vert \tilde{i})\\
    &=\text{Tr}\left(\sum_{i,k}q_O^i(\tilde{k}\vert k)E_k^iq_M(i\vert \tilde{i})\rho^j\right)\\
    &=\text{Tr}\left(E_{\tilde{k}}^{\tilde{i}}\rho^j\right) = p(\tilde{k} | M_{\tilde{i}}, P_j),
\end{align*}
where $E_{\tilde{k}}^{\tilde{i}}$ are the POVM elements of the measurement procedures in $\mathbb{M}^Q_{\rm si}$, discussed in the Appendix~\ref{sec:appendix1}. Therefore we might access the quantum contextual behavior from $\mathbb{B}_{\rm si}$ that is maximally quantum contextual~\cite{spekkens2005contextuality,galvao2002foundations}. Since $T$ is a free operation, the specific quantum realization  we used in the domain $\mathbb{B}$ cannot be noncontextual. 
\end{proof}

Applying the same reasoning towards preparation procedures, a generalization follows:

\begin{exm}\label{example:BtowardsSi_MP}
Consider scenarios of the kind $\mathbb{B}:= (|J|,|I|,2,\mathbb{E}_{\mathbb{P}},\emptyset),$  with $|J|$ even, $|J|\geq 4$ and  $|I| \geq 2$, and $\mathbb{E}_{\mathbb{P}}=\mathbb{E}_{{\mathbb{P}},\rm si}\cup \mathbb{E}'$, where $\mathbb{E}'$ does not involve the first four preparations. There is always a quantum behavior $B \in \mathbb{B}$ and a free operation $T\in \mathcal{F}$, with image $T(\mathbb{B}) = \mathbb{B}_{\rm si}$ the simplest scenario such that   $T(B) \in \mathbb{B}_{\rm si}$ is a quantum  contextual behavior. 
\end{exm}
 It is clear that there exists a quantum contextual correlation for such a scenario since we can consider the same quantum contextual behavior from example~\ref{example: BtowardsSi_M}, and complete the procedures with anything such that the equivalences $\mathbb{E}'$ do not involve the first four procedures. Then,  there will certainly exist some pre-processings from these towards the preparations \eqref{preparation 11}-\eqref{preparation --}, with the same description as the one given by Corollary~\ref{corollary: trivialSim}. 
 
\subsubsection{Using the composition}


In this section we take advantage of the consequences of Theorem~\ref{theo:block preserve resource} to witnessing quantum contextuality. Namely,
\begin{coro}
Consider a behavior $B_1 \boxplus B_2$, with $B_1\in\mathbb{B}_1$ and $B_2\in\mathbb{B}_2$, that has some quantum realizations and is contextual. Then, $B_1$ or $B_2$ must be contextual. Mathematically,
\begin{align}
    &B_1 \boxplus B_2 \in QC(\mathbb{B}_1 \boxplus \mathbb{B}_2) \\
    &\iff B_1 \in QC(\mathbb{B}_1)\,\, or\,\, B_2 \in QC(\mathbb{B}_2).
\end{align}
Where $QC(\mathbb{B})$ is the set of contextual points in the scenario that have some quantum realization.
\end{coro}
Hence, using the composition $\boxplus$, we are  constructing higher-dimensional polytopes that  will inherit quantum contextuality from its lower-dimensional components. Thus, whenever very complex scenarios can be understood as the product of lower-dimensional ones, we can find the arising noncontextual polytope structure from the product elements and build up quantum violations from those in the components. We shall see examples of how this can be done.
First, let us  give an intuitive geometric view of why complex scenarios acquire quantum contextuality from simple ones. Let \begin{tikzpicture}
\draw[blue] (1,0)--(1.5,0);
\draw[fill] (0,0) circle[radius=0.1] ;
\draw (0,0)--(1,0);
\draw[fill] (1,0) circle[radius=0.1];
\draw[fill](1.5,0)--(2,0);
\draw[fill] (2,0) circle[radius=0.1];
\draw[fill=blue] (1.5,0) circle[radius=0.1];
\end{tikzpicture} be the pictorial representation of a prepare-and-measure experimental scenario, meaning that this is some convex polytope $\mathbb{B}$ with the quantum set containing the left black line and the one in blue. Then, the product $\mathbb{B}\boxplus \mathbb{B}$ between two of these 1-dimensional convex polytopes will be such as represented in Fig.~\ref{Bloch para med}.

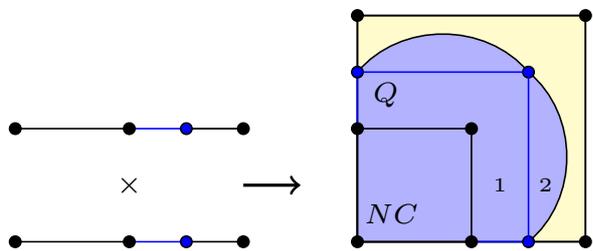
\begin{figure}[H]
\centering
\scalebox{1.5}{
\begin{tikzpicture}
\draw[fill=yellow!25] (0,0)--(0,2)--(2,2)--(2,0)--(0,0);
\coordinate   (A) at (1.5,0) ;
\coordinate  (B) at (0,1.5) ;
\coordinate (C) at (1.5,1.5);
\draw[fill=blue!30] (A) to [bend right=50]  (C) to [bend right=50] (B)--(0,0)--(A);
\draw[blue] (0,1)--(0,1.5);
\draw[blue] (1,0)--(1.5,0);
\draw[fill=blue!30] (0,0)--(1,0)--(1,1)--(0,1)--(0,0);
\draw[fill] (0,0) circle[radius=0.05] ;
\draw (0,0)--(1,0);
\draw[fill] (1,0) circle[radius=0.05];
\draw (2,0)--(2,2)--(0,2);
\draw[fill=blue] (1.5,1.5) circle[radius=0.05];
\draw[fill](1.5,0)--(2,0);
\draw[fill] (2,0) circle[radius=0.05];
\draw[fill] (2,2) circle[radius=0.05];
\draw[fill] (1,1) circle[radius=0.05];
\draw[fill=blue] (1.5,0) circle[radius=0.05];
\draw (0,0)--(0,1);
\draw[fill] (0,1) circle[radius=0.05];
\draw[fill](0,1.5)--(0,2);
\draw[fill] (0,2) circle[radius=0.05];
\draw[fill=blue] (0,1.5) circle[radius=0.05];

\draw[blue](1.5,0)--(1.5,1.5)--(0,1.5);

\node at (0.3,0.25) {\scriptsize $NC$};
\node at (0.25,1.3) {\scriptsize $Q$};
\node at (1.25,0.5) {\tiny $1$};
\node at (1.65,0.5) {\tiny $2$};

\draw[blue] (-2,1)--(-1.5,1);
\draw[fill] (-3,1) circle[radius=0.05] ;
\draw (-3,1)--(-2,1);
\draw[fill] (-2,1) circle[radius=0.05];
\draw[fill](-1.5,1)--(-1,1);
\draw[fill] (-1,1) circle[radius=0.05];
\draw[fill=blue] (-1.5,1) circle[radius=0.05];

\draw[blue] (-2,0)--(-1.5,0);
\draw[fill] (-3,0) circle[radius=0.05] ;
\draw (-3,0)--(-2,0);
\draw[fill] (-2,0) circle[radius=0.05];
\draw[fill](-1.5,0)--(-1,0);
\draw[fill] (-1,0) circle[radius=0.05];
\draw[fill=blue] (-1.5,0) circle[radius=0.05];

\node at (-2,0.5) {\scriptsize $\times$};
\draw[thick,->] (-1,0.5)--(-0.5,0.5);
\end{tikzpicture}
}
\caption{\label{Bloch para med}\textbf{Representation of the polytope structure arising from the product scenario.} We stress that it is not clear how should be the new form of the \textit{quantum} set $Q(\mathbb{B})$, even though it is clear the polytope structure for both $NC$ and the larger polytope of statistics. In this picture we have used the fact that $NC(\mathbb{B}) \subseteq Q(\mathbb{B})$. From the convex nature of $Q(\mathbb{B})$, the product scenario must have a quantum contextual set that is at least of the form given by region $1$, but it could also be given by region $2$ and the study of maximal violations for noncontextuality inequalities shall answer such questions, see Refs.~\cite{ambainis2016parity,chailloux2016optimalbounds}.}
\end{figure}
From Fig.~\ref{Bloch para med}, we see that if one constructs very complex scenarios, when they are associated with decompositions of $\mathbb{B}^{\boxplus n}$ for $n$ as large as we could imagine, the resource is always present.

The arguably simplest construction one might consider is to take sequential products of the simplest scenario, $\mathbb{B}:=\mathbb{B}^{\boxplus n}_{\rm si}$. Interestingly, using some symmetry arguments and the tight noncontextuality inequalities of the simplest scenario, we can obtain the following result for such construction:

\begin{lemma}
\label{lemma:general violations}
For any scenario of the form  $\mathbb{B}:=\mathbb{B}_{\rm si}^{\boxplus n}$, $n\geq 1$, every tight and nontrivial noncontextuality inequality will be violated by some  quantum contextual behavior. 
\end{lemma}
We provide a proof in  Appendix~\ref{appenidx: symmetry}. Therefore, for these scenarios, there exist quantum contextual correlations \emph{with respect to all (nontrivial) noncontextuality inequalities that define the polytope} $NC(\mathbb{B})$. A resource theoretic consequence arises as a corollary from such lemma.

\begin{coro}[Quantum advantages for $\mathbb{B}=\mathbb{B}_{\rm si}^{\boxplus n}$]
\label{corollary: QAdvantagesfromBsi}
Consider a quantum information task that has a success rate defined by a function $g: \mathbb{B} \to \mathbb{R}_+$, for $\mathbb{B}$ of the form of $\mathbb{B}_{\rm si}^{\boxplus n}$, such that the noncontextual bound for the success, $g^{NC}(B) \leq \delta$, for some $\delta \in \mathbb{R}$, can be expressed as a linear combination of the noncontextuality inequalities of $NC(\mathbb{B})$. Then, there exists a quantum behavior $B^Q$ such that $g(B^Q) > \delta$.
\end{coro}
Thus, Lemma~\ref{lemma:general violations} represents a general proof of quantum \emph{advantage} in tasks related to scenarios of the form $\mathbb{B}_{\rm si}^{\boxplus n}$, for $n\geq 2$ (whenever the success rate of the operational task is defined by a function $g$ that is a convex-linear function of the noncontextuality inequalities $NC(\mathbb{B}_{\rm si})$).

As another notable example, we discuss the quantum cloning scenario from Ref.~\cite{lostaglio2019contextual} which we name as $\mathbb{B}_{\rm qc}$ and define in  Appendix~\ref{appenidx: blocks}. $\mathbb{B}_{\rm qc}$ is an operational scenario that reproduces the statistics for an important quantum task known as state-dependent quantum cloning, in which contextuality underpins quantum advantage~\cite{lostaglio2019contextual}. 


\begin{exm}[Quantum cloning inherits contextuality from $\mathbb{B}_6$]
\label{example:Bqc}

The scenario $\mathbb{B}_{\rm qc}$ related to the state-dependent quantum cloning task, can be written as
\begin{equation}
\label{eq:quantum cloning pieces}
    \mathbb{B}_{\rm qc} = \mathbb{B}_{6}\boxplus \mathbb{B}_{6} \boxplus \mathbb{B}_{6} ,
\end{equation}
where $\mathbb{B}_6 := (4,6,2,\mathbb{E}_{\mathbb{P},\rm si},\emptyset)$. The \textit{inner} polytope structure of $\mathbb{B}_{\rm qc}$ is then given by $\mathbb{B}_{6}$.
\end{exm}
 This example allow us to conclude the following: Eq.~\eqref{eq:quantum cloning pieces} is a proof of quantum contextuality in the quantum cloning scenario, since $\mathbb{B}_6$ has quantum contextual behaviors. The first point is in agreement with Refs.~\cite{lostaglio2019contextual,sainz2020non} providing a new understanding  of the advantage in the cloning scenario, i.e. in terms of quantum contextuality present in the smaller $\mathbb{B}_6$.

\section{Discussion} 
\label{sec:conclusions}

In this work, we use the resource theory of contextuality introduced in Ref.~\cite{duarte2018resource}. We examine, with a resource theoretic perspective, the preservation of contextuality due to measurement simulability; known techniques used in experimental tests of contextuality; and the polytope structure of novel composed scenarios. These composed scenarios can be interpreted as a strategy for extending inequalities from smaller scenarios to larger ones. 

To elaborate, we establish a connection between the simulation of measurements in operational theories with free operations, and the creation of simpler scenarios achieved by omitting certain measurements. Furthermore, by recognizing that mixing is formally a free operation, we offer a new interpretation of the engineering of operational equivalences used in the tests detailed in Ref.~\cite{mazurek2016experimental}. We conclude that the free operations provide a simple and rigorous argument in favor of the experimental conclusions drawn from Ref.~\cite{mazurek2016experimental} regarding the use of secondary procedures as a tool for witnessing contextuality from imperfect correlations.  

Moreover, we introduce a composition of scenarios allowing the construction of complex scenarios while conserving the resource. This brings light to the importance of the resource theory developed in Ref.~\cite{duarte2018resource},  specially in situations limited by the intrinsic complexity of numerically studying correlation polytopes.

We then apply the techniques to analyze \emph{quantum} contextuality, leading to both foundational and practical implications. We have demonstrated that there always exists quantum contextual behaviors for a class of prepare-and-measure scenarios -- and such a class encompasses the scenarios from Refs.~\cite{schmid2018contextual,lostaglio2019contextual,Pusey2018simplest,spekkens2009preparation}. We also show that quantum contextuality is present for every nontrivial facet of the noncontextual polytope for scenarios of the form $\mathbb{B}_{\rm si}^{\boxplus n}$. Moreover, we show that the scenario related to the task of state-dependent cloning, $\mathbb{B}_{\rm qc}$, can be decomposed in terms of the simpler scenario $\mathbb{B}_6$. Thus, we can conclude that the quantum resource present in $\mathbb{B}_{\rm qc}$ is inherited from quantum contextuality in $\mathbb{B}_6$. This also allows one to understand $NC(\mathbb{B}_{\rm qc})$ via the inequalities of $NC(\mathbb{B}_6)$, which is a much simpler computational task. 

\subsection{Relation with previous work}

Resource theoretic investigations of prepare-and-measure  (generalized) contextuality scenarios remain largely unexplored. In this work, we have presented various qualitative and quantitative results that analyze the specific structure of these scenarios. We believe these findings can be impactful in practical applications, enabling the development of novel demonstrations of quantum contextuality, or the lack thereof, akin to the successful application of the monotone $\mathsf{d}$, as seen in Ref.~\cite{baldi2021emergence}, for constraining the emergence of noncontextuality under quantum Darwinism.

On the contrary, for a different yet related notion of contextuality, namely Kochen--Specker (KS) contextuality, not only there is a vast literature devoted to formalize a resource theory of it~\cite{grudka2014quantifying,horodecki2015axiomatic,abramsky2017contextual,abramsky2019comonadic,amaral2019resource,li2020robustness,horodecki2023rank}, but also its application to relevant tasks~\cite{frembs2018contextuality,karanjai2018contextuality,karvonen2021neither}. In this case, very similar constructions, such as the operation $\boxplus$ we have considered here, have been introduced for empirical models in the resource theory of KS-contextuality in Ref.~\cite{abramsky2019comonadic}.  Our description differs from theirs since it uses a fairly \textit{different} notion of nonclassicality, namely,  generalized contextuality. This fundamental difference not only conceptualizes the two notions distinctly but also results in a completely different scenario description. Consequently, there is no compelling argument for why results from generalized contextuality scenarios should apply to the measurement scenarios in the KS formalism.

\subsection{Further directions}
  
The examples here discussed are far from exhausting the possibilities of the tools developed. Therefore, finding other physically appealing scenarios, or introducing novel composition rules beyond $\boxplus$ is an interesting perspective. Further investigations could also use small perturbations over the quantum measurements~\cite{d2005classical}, and with the help of contextuality monotones, try to understand  how much can one perturb the quantum behaviors and still witness contextuality, which is  important for experimental implementation of generalized noncontextuality.

One potential venue of exploration involves studying the interplay between the composition rule $\boxplus$ and scenarios relevant for quantum computation. For instance, it remains unclear the relevance of generalized contextuality in some measurement based schemes of quantum computation~\cite{braviy2016trading,bartolucci2023fusion}. Possibly, each measurement step leads to some defined scenario $\mathbb{B}$, and the sequence of measurements leading to some sequence $\mathbb{B}_1 \boxplus \dots \boxplus \mathbb{B}_n$. Furthermore, it is likely that the resource theory of generalized contextuality will play a leading role in explaining hardness of classically simulating quantum computations for the so-called $\Lambda$-polytope method~\cite{zurel2020hidden,zurel2023simulating,cihan2021ontheextremal}, that provides an ontological model that is measurement noncontextual, but preparation contextual, and for which it is unknown what nonclassical resources drive the simulation overhead.

To conclude, our work motivates the utilization of resource-theoretic analysis  of generalized contextuality, as we provide new qualitative insights into the underlying polytope structure within prepare-and-measure scenarios. We have demonstrated the applicability of the composition rule $\boxplus$, that can be interpreted as an instance of a lifting of noncontextuality inequalities.  Formally, lifting has been successfully employed in analyzing other pertinent quantum information polytopes, such as Bell inequalities that define facets in local polytopes~\cite{pironio2005lifting}. Nevertheless, the formal lifting strategies for noncontextuality inequalities in prepare-and-measure scenarios remain an uncharted territory, lacking a formal complete description. We hope our work paves the way for future investigations,  encouraging exploration in this direction. \newline 

\section*{Acknowledgements} \label{sec:acknowledgements}

The authors would like to thank Ana Belén Sainz and Marcelo Terra Cunha for helpful comments on an early version of this work.

RW,  AT and BA would like to thank National Council for the Improvement of Higher Education (CAPES)  and National Council for Scientific and Technological Development (CNPq)  for their support during this research. RDB would like to thank S\~{a}o Paulo Research Foundation (Fapesp), that supported this work via projects numbers 2016/24162-8 and 2019/02221-0.  BA was supported financially by   S\~{a}o Paulo Research Foundation, Aux\'ilio \`a Pesquisa - Jovem Pesquisador, grant number $2020/06454-7$, and Instituto Serrapilheira, Chamada 2020.

\bibliography{bibliography}

\begin{appendix}
\onecolumngrid
\newpage
\section{Noncontextual polytope for the simplest scenario} \label{sec:appendix2}

With the methods developed by Ref.~\cite{schmid2018all} it is possible to fully characterize the noncontextual polytope for the scenario $\mathbb{B}_{\rm si}$. For a given behavior $B \in \mathbb{B}_{\rm si}$ using the shorthand for $p(1\vert M_i, P_j) = p_{ij}$ we have that the facets of the noncontextual polytope are tightly characterized by the following set of inequalities:
\begin{align}
    0 \leq p_{ij} \leq 1, \,\, \forall M_i, P_j \label{simplest equations 1}\\
    p_{12}+p_{22}-p_{14}-p_{23}\leq 1\\
    p_{12}+p_{22}-p_{13}-p_{24}\leq 1\\
    p_{22}+p_{13}-p_{12}-p_{24}\leq 1\\
    p_{12}+p_{23}-p_{22}-p_{14}\leq 1\\
    p_{22}+p_{14}-p_{12}-p_{23}\leq 1\\
    p_{23}+p_{14}-p_{12}-p_{22}\leq 1\label{correction to all}\\
    p_{12}+p_{24}-p_{22}-p_{13}\leq 1\label{counter example equation}\\
    p_{13}+p_{24}-p_{22}-p_{12}\leq 1\label{simplest equations 9}
\end{align}  
Taken from Refs.~\cite{schmid2018all} and~\cite{eu2020mestrado}.

\section{Proof of theorem~\ref{theo:block preserve resource}}\label{appenidx: blocks}


\begin{dem}

Let $B_1 \in NC(\mathbb{B}_1),B_2 \in NC(\mathbb{B}_2) $. Hence, there are $(\Sigma^{(i)}, \Lambda^{(i)}, \Pi^{(i)}, \Theta^{(i)})$ where $\Pi^{(i)}$ and $\Theta^{(i)}$, $i=1,2$, respect the operational equivalences at the ontological model level respectively for each scenario. For sets of labels we define $K_i, I_i, J_i, \{a_i\}, \{b_i\}$ as before (see Def.~\ref{def: box scenario}), for their respective operational primitives from $\mathbb{B}_i$. The scenarios are finite and the operational equivalences are fixed and finite as well, so each set ranges over a finite set of labels.

\begin{align*}
    &K := K_1 \cup K_2 \\
    &I := I_1 \cup I_2 \\
    &J := J_1 \cup J_2 \\
    &\{a\}_{a=1}^{|\mathbb{E}_{\mathbb{P}_1 \cup \mathbb{P}_2}|} := \{a_1\}_{a_1=1}^{|\mathbb{E}_{\mathbb{P}_1}|} \cup \{a_2\}_{a_2=1}^{|\mathbb{E}_{\mathbb{P}_2}|}, \\
    &\{b\}_{b=1}^{|\mathbb{E}_{\mathbb{M}_1 \cup \mathbb{M}_2}|} := \{b_1\}_{b_1=1}^{|\mathbb{E}_{\mathbb{M}_1}|} \cup \{b_2\}_{b_2=1}^{|\mathbb{E}_{\mathbb{M}_2}|}.
\end{align*}

From the definition of noncontextuality at the ontological model level, we have the equations, here and throughout this appendix we simplify the notation to $p(k|M_i,P_j) \equiv p(k|i,j)$.

\begin{align}
   & p(k_1|i_1,j_1) = \sum_{\lambda_1 \in \Lambda_1}\xi_{[k_1|i_1]}(\lambda_1)\mu_{j_1}(\lambda_1), \\ 
    &p(k_2|i_2,j_2) = \sum_{\lambda_2\in\Lambda_2}\xi_{[k_2|i_2]}(\lambda_2)\mu_{j_2}(\lambda_2), \\
    &\label{eq op 1}\sum_{j_1}(\alpha_{j_1}^{a_1}-\beta_{j_1}^{a_1})\mu_{j_1}(\Omega_1) = 0, \,\,\, \forall  a_1, \forall \Omega_1 \in \Sigma_1 \\
    &\label{eq op 2}\sum_{j_2}(\alpha_{j_2}^{a_2}-\beta_{j_2}^{a_2})\mu_{j_1}(\Omega_2) = 0,\,\,\, \forall  a_2, \forall \Omega_2 \in \Sigma_2 \\
    &\sum_{k_1,i_1}(\alpha_{[k_1|i_1]}^{b_1}-\beta_{[k_1|i_1]}^{b_1})\xi_{[k_1|i_1]}(\lambda_1) = 0,\,\,\, \forall \lambda_1,b_1, \label{eq op 5}\\
    &\sum_{k_2,i_2}(\alpha_{[k_2|i_2]}^{b_2}-\beta_{[k_2|i_2]}^{b_2})\xi_{[k_2|i_2]}(\lambda_2) = 0,\,\,\, \forall \lambda_2,b_2. \label{eq op 6}
\end{align}
In \eqref{eq op 1}-\eqref{eq op 6}, $\xi_{[\cdot | i_1]}\in \Theta^{(1)}, \mu_{j_1} \in \Pi^{(1)}$, and similarly for the remaining distributions. If we consider the product between the behaviors, $B_1 \boxplus B_2$, we can construct a novel ontological model using as the ontic space $\Lambda := \Lambda^{(1)} \sqcup \Lambda^{(2)}$ the disjoint union between the two sets. We then define $\tilde{\xi}_{[k|i]}: \Lambda \to [0,1]$ as,
    \begin{equation*}
    \tilde{\xi}_{[k|i]}(\lambda) := \left\{ \begin{matrix}
    \xi_{[k_1|i_1]}(\lambda_1), & \lambda = (\lambda_1,1)  \\
    \xi_{[k_2|i_2]}(\lambda_2), & \lambda = (\lambda_2,2)
    \end{matrix}\right.
\end{equation*}


 \begin{align*}
        \forall \lambda \in \Lambda,  \sum_{k \in K} \tilde{\xi}_{[k|i]}(\lambda) =  \left\{ \begin{matrix}
    \sum_{k\in K}\tilde{\xi}_{[k|i_1]}(\lambda), \text{ if }i\in I_1\\
    \sum_{k\in K}\tilde{\xi}_{[k|i_2]}(\lambda), \text{ if }i \in I_2
    \end{matrix}\right. = \left\{ \begin{matrix}
    \sum_{k_1\in K_1}\tilde{\xi}_{[k_1|i_1]}(\lambda), \text{ if }i\in I_1\\
    \sum_{k_1\in K_1}\tilde{\xi}_{[k_1|i_2]}(\lambda), \text{ if }i \in I_2
    \end{matrix}\right. = \left\{ \begin{matrix}
    1,\\
    1
    \end{matrix}\right. = 1
    \end{align*}

So that the extended functions are normalized in the ontic space $\Lambda$. We have considered that, whenever $i \in I_1 \cap I_2$ any function $\xi_{[k_1|i_1]}$ or $\xi_{[k_2|i_2]}$ will serve, we then just need to pick one and use it for our noncontextual ontological model. This means that if we have two scenarios with the same procedures, $\{M_1,M_2\},\{M_1,M_2\} \to \{M_{1_1},M_{2_1},M_{1_2},M_{2_2}\} \equiv \{M_1,M_2\}$. Therefore we can recognize if two procedures are just simply the same. In this sense, we can have that the number of procedures in $\mathbb{B}_{\rm si}$ and $\mathbb{B}_{\rm si}^{\boxplus n}$ are the same, so that we simplify the scenario's description. 

For $\tilde{\mu}_j$, the ontic spaces are finite and we write $\tilde{\mu}_j(\{\lambda\}) \equiv \tilde{\mu}_j(\lambda)$. Let $j \in J$, we define that, if $j \in J_1,$

 \begin{equation*}
    \tilde{\mu}_j(\lambda) := \left\{ \begin{matrix}
    \mu_{j_1}(\lambda), & \text{ if } \lambda = (\lambda_1,1)\\
    0, & \text{ if } \lambda = (\lambda_2,2)
    \end{matrix}\right. 
\end{equation*}
and similarly if $j \in J_2$. Again, when $j \in J_1 \cap J_2$ we choose one of the ontological descriptions as our fixed definition for the preparation procedure associated with it.  With this definition we have that, for any $j \in J$,

\begin{equation*}
    \sum_{\lambda \in \Lambda}\tilde{\mu}_j(\lambda) = \sum_{\lambda=(\lambda_1,1) \in \Lambda}\tilde{\mu}_j(\lambda) +  \sum_{\lambda=(\lambda_2,2) \in \Lambda}\tilde{\mu}_j(\lambda) 
\end{equation*}
and, whenever $j \in J_1$ or $j \in J_2$ we recover the normalization condition from the already defined distributions in the parts. We then obtain that any $p(k|i,j) $ in $B_1 \boxplus B_2$  will have an ontological description,

\begin{align*}
    & \sum_{\lambda \in \Lambda_1 \sqcup \Lambda_2}\tilde{\xi}_{[k|i]}(\lambda)\tilde{\mu}_{j}(\lambda) = \left \{ \begin{matrix}
    \sum_{\lambda \in \Lambda}\tilde{\xi}_{[k_1|i_1]}(\lambda)\tilde{\mu}_{j_1}(\lambda), & \text{ if }k,i,j\in K_1,I_1,J_1 \\
    \sum_{\lambda \in \Lambda}\tilde{\xi}_{[k_2|i_2]}(\lambda)\tilde{\mu}_{j_2}(\lambda), & \text{ if } k,i,j\in K_2,I_2,J_2
    \end{matrix}\right.\\
    &=\left \{ \begin{matrix}
    \sum_{\lambda_1 \in \Lambda_1}\xi_{[k_1|i_1]}(\lambda_1)\mu_{j_1}(\lambda_1), & \text{ if } k,i,j\in K_1,I_1,J_1  \\
    \sum_{\lambda_2 \in \Lambda_2}\xi_{[k_2|i_2]}(\lambda_2)\mu_{j_2}(\lambda_2), & \text{ if }k,i,j\in K_2,I_2,J_2
    \end{matrix}\right. = \left \{ \begin{matrix}
    p(k_1|i_1,j_1), & \text{ if } k,i,j\in K_1,I_1,J_1  \\
    p(k_2|i_2,j_2), & \text{ if }k,i,j\in K_2,I_2,J_2
    \end{matrix}\right. = B_1 \boxplus B_2
\end{align*}

 Notice that in the new scenario $\mathbb{B}_1 \boxplus \mathbb{B}_2$ it is at play our operational constraint that the preparations of the parts do not interact with the measurements of one another. The operational equivalences defined in the scenario $\mathbb{B}_1 \boxplus \mathbb{B}_2$ are the ones from \eqref{op eq box scenario}, so we need to study the following objects:
\begin{equation}
    \sum_{j}(\alpha_{j}^a-\beta_j^a)\tilde{\mu}_j(\lambda), \,\,\, \forall \lambda \in \Lambda_1 \sqcup \Lambda_2 , \forall \bm{\gamma}^a \in  \mathbb{E}_{\mathbb{P}_1 \cup \mathbb{P}_2},
\end{equation}
It will be true, for all $\lambda \in \Lambda_1 \sqcup \Lambda_2$, the following,

\begin{align*}
    \sum_{j}(\alpha_{j}^{a}-\beta_j^{a})\tilde{\mu}_j(\lambda) &= \underbrace{ \sum_{j_1}(\alpha_{j_1}^{a}-\beta_{j_1}^{a})\tilde{\mu}_{j_1}(\lambda)}_{\stackrel{\eqref{eq op 1}}{=}0} + \underbrace{\sum_{j_2}(\alpha_{j_2}^{a}-\beta_{j_2}^{a})\tilde{\mu}_{j_2}(\lambda)}_{\stackrel{\eqref{op eq box scenario}}{=}0} = 0, \forall \bm{\gamma}^a_{\mathbb{P}_1\cup\mathbb{P}_2} = \bm{\gamma}^{a_1}_{\mathbb{P}_1\cup\mathbb{P}_2}
\end{align*}

\begin{align*}
    \sum_{j}(\alpha_{j}^{a}-\beta_j^{a})\tilde{\mu}_j(\lambda) &= \underbrace{ \sum_{j_1}(\alpha_{j_1}^{a}-\beta_{j_1}^{a})\tilde{\mu}_{j_1}(\lambda)}_{\stackrel{\eqref{op eq box scenario}}{=}0} + \underbrace{\sum_{j_2}(\alpha_{j_2}^{a}-\beta_{j_2}^{a})\tilde{\mu}_{j_2}(\lambda)}_{\stackrel{\eqref{eq op 2}}{=}0} = 0, \forall \bm{\gamma}^a_{\mathbb{P}_1\cup\mathbb{P}_2} = \bm{\gamma}^{a_2}_{\mathbb{P}_1\cup\mathbb{P}_2}
\end{align*}

And for all $\lambda \in \Lambda_1 \sqcup \Lambda_2$ we also have, $\forall \bm{\gamma}^b_{\mathbb{M}_1\cup\mathbb{M}_2} = \bm{\gamma}^{b_1}_{\mathbb{M}_1\cup\mathbb{M}_2}$, 

\begin{equation*}
    \sum_{k,i}(\alpha_{[k|i]}^b-\beta_{[k|i]}^b)\tilde{\xi}_{[k|i]}(\lambda) = \underbrace{\sum_{k_1,i_1}(\alpha_{[k_1|i_1]}^{b_1}-\beta_{[k_1|i_1]}^{b_1})\tilde{\xi}_{[k_1|i_1]}(\lambda_1)}_{\stackrel{\eqref{eq op 5}}{=}0} + \underbrace{\sum_{k_2,i_2}(\alpha_{[k_2|i_2]}^{b_1}-\beta_{[k_2|i_2]}^{b_1})\tilde{\xi}_{[k_2|i_2]}(\lambda_2)}_{\stackrel{\eqref{op eq box scenario}}{=}0} = 0,
\end{equation*}
and $ \forall \bm{\gamma}^b_{\mathbb{M}_1\cup\mathbb{M}_2} = \bm{\gamma}^{b_2}_{\mathbb{M}_1\cup\mathbb{M}_2}$, 
\begin{equation*}
    \sum_{k,i}(\alpha_{[k|i]}^b-\beta_{[k|i]}^b)\tilde{\xi}_{[k|i]}(\lambda) = \underbrace{\sum_{k_1,i_1}(\alpha_{[k_1|i_1]}^{b_2}-\beta_{[k_1|i_1]}^{b_2})\tilde{\xi}_{[k_1|i_1]}(\lambda_1)}_{\stackrel{\eqref{op eq box scenario}}{=}0} + \underbrace{\sum_{k_2,i_2}(\alpha_{[k_2|i_2]}^{b_1}-\beta_{[k_2|i_2]}^{b_1})\tilde{\xi}_{[k_2|i_2]}(\lambda_2)}_{\stackrel{\eqref{eq op 6}}{=}0} = 0.
\end{equation*}

This proves that the ontological model constructed  is noncontextual for the behavior $B_1 \boxplus B_2$ whenever $B_1,B_2$ are also noncontextual behaviors. 

For the $(\Leftarrow)$ part of the proof, suppose that the behavior $B_1 \boxplus B_2$ has a noncontextual ontological model $(\Sigma, \Lambda, \Pi, \Theta)$. We know that this $\mathbb{B}_1\boxplus \mathbb{B}_2$ scenario has the same operational equivalences as both the scenarios $\mathbb{B}_1$ and $\mathbb{B}_2$ divided, by means of the weight vectors, e.g.,    $(\alpha_1^{a_1}, \alpha_2^{a_1}, \dots, \alpha_{j_1}^{a_1}, 0, \dots, 0)$. Hence, there exists an ontological model for $B_1$ inherited from $B_1 \boxplus B_2$ using the operational equivalences:
\begin{equation*}
    \sum_j(\alpha_j^a-\beta_j^a)\mu_j(\lambda) = 0, \forall \lambda \implies \sum_{j_1} (\alpha_{j_1}^a-\beta_{j_1}^a)\mu_{j_1}(\lambda) = 0, \forall \lambda, \forall \bm{\gamma}^a_{\mathbb{P}_1\cup\mathbb{P}_2} \in \mathbb{E}_{\mathbb{P}_1\cup\mathbb{P}_2},
\end{equation*}
where we can restrict $\{a\}$ to some set of labels $\{a_1\}$ and reduced vectors $\bm{\gamma}_{\mathbb{P}}^{a_1}$ by cutting the zeros. We get the same for the behavior $B_2$. The ontological description of the probabilities we get immediately:
\begin{equation}
    p(k_1|i_1,j_1) := \sum_{\lambda \in \Lambda} \xi_{[k=k_1|i=i_1]}(\lambda)\mu_{j=j_1}(\lambda) ,
\end{equation}
for any $k_1,i_1,j_1 \in K_1,I_1,J_1$, and similarly,
\begin{equation}
    p(k_2|i_2,j_2) := \sum_{\lambda \in \Lambda} \xi_{[k=k_2|i=i_2]}(\lambda)\mu_{j=j_2}(\lambda).
\end{equation}

\end{dem}

\begin{exm}
Let us consider, as an example, the scenario from Ref.~\cite{lostaglio2019contextual}. The operational scenario can be described in the following terms: There are six preparations, $P_a,P_b$ associated with the preparation of the states that we want to clone. The preparations $P_\alpha, P_\beta$ correspond to procedures $P_a,P_b$ that have passed through a cloning machine, in other words, outputs of the cloning transformation (considered preparations on their own). Last but now least, we consider the preparations associated with having cloned the states $P_{aa}, P_{bb}$. We also include the orthogonal counterparts $P_{s^\perp}$, meaning that the probabilities are impossible to distinguish by one-shot measurements, in complete analogy with the prepare-and-measure scenarios from Ref.~\cite{schmid2018contextual}, in order to generate the operational equivalences necessary for generalized contextuality. This means that we have a set $\mathbb{P}:=\{P_s,P_{s^\perp} \,\, \vert \,\, s\in \{a,b,\alpha,\beta,aa,bb\}\}$ and the operational equivalences over the preparation procedures as defined by:
\begin{align}
    &\frac{1}{2}P_{a} + \frac{1}{2}P_{a^\perp} \simeq  \frac{1}{2}P_{b} + \frac{1}{2}P_{b^\perp} \label{cloning op equivalences 1}\\
    & \frac{1}{2}P_{\alpha} + \frac{1}{2}P_{\alpha^\perp}  \simeq  \frac{1}{2}P_{aa} + \frac{1}{2}P_{aa^\perp} \label{cloning op equivalences 2}\\
    &  \frac{1}{2}P_{\beta} + \frac{1}{2}P_{\beta^\perp} \simeq  \frac{1}{2}P_{bb} + \frac{1}{2}P_{bb^\perp}\label{cloning op equivalences 3} 
\end{align}
for the measurement procedures we have the six binary-outcome ones $M_s, s\in \{a,b,\alpha,\beta,aa,bb\}\}$.  Defining such scenario by $\mathbb{B}_{\rm qc}$ we notice that we can write, using the definition of $\boxplus$ from before, that
\begin{equation}\label{cloning scenario equality}
    \mathbb{B}_{\rm qc} = \mathbb{B}_{6}^{\boxplus 3}.
\end{equation}
Here, the scenario $\mathbb{B}_6$ is the scenario with the same operational structure of the scenario $\mathbb{B}_{\rm si}$ but with $6$ measurement procedures.  We needed to add some symmetries associated with the fact that the scenario $\mathbb{B}_6^{\boxplus 3}$ would have $18$ measurement procedures,  therefore making notice of the symmetry $I_1 = I_2 = I_3 \implies I := I_1 \cup I_2 \cup I_3 = I_1 = \{a,b,\alpha,\beta,aa,bb\}$. And also that $K_1 = K_2 = K_3$. Under these circumstances the equation \eqref{cloning scenario equality} between scenarios holds. This example is also helpful to discuss some aspects of the map $\boxplus$. With this consideration we then get that,
\begin{equation}\label{noncontextuality polytopes product}
    NC(\mathbb{B}_{\rm qc}) = NC(\mathbb{B}_6) \times NC(\mathbb{B}_6) \times NC(\mathbb{B}_6). 
\end{equation}

\end{exm}
 
\subsection*{\texorpdfstring{$l_1$}{l1}-distance}

We prove Lemma~\ref{lemma:Compositionl1-distance} as follows: Recall that the $l_1$-distance monotone is defined as,
\begin{equation}\label{eq: appendix def l1 distance}
    \mathsf{d}(B) := \min_{B^* \in NC(\mathbb{B})}\max_{i,j}\sum_k \vert p(k|i,j) - p^*(k|i,j) \vert .
\end{equation}
Let $B_1^*$ and $B_2^*$ be the noncontextual behaviors achieving the minimum in equation \eqref{eq: appendix def l1 distance}. Denoting $(p^*(k|i,j))_{k\in K, i\in I, j\in J} \equiv B_1^*\boxplus B_2^*$, 
\begin{align*}
    \mathsf{d}(B_1 \boxplus B_2) &\leq  \max_{i,j}\sum_k \vert p(k|i,j) - p^*(k|i,j)  \vert \\
    &= \max_{i,j} \left \{ \left\{ \sum_k \vert p_1(k_1|i_1,j_1) - p_1^*(k_1|i_1,j_1)  \vert \right\}\cup \left\{ \sum_k \vert p_2(k_2|i_2,j_2) - p_2^*(k_2|i_2,j_2)  \vert \right \} \right\} \\
    &=\max \left \{ \max_{i_1,j_1}\sum_k \vert p_1(k_1|i_1,j_1) - p_1^*(k_1|i_1,j_1)  \vert , \max_{i_2,j_2}\sum_k \vert p_2(k_2|i_2,j_2) - p_2^*(k_2|i_2,j_2)  \vert \right\}\\
    &= \max\left \{ \mathsf{d}(B_1), \mathsf{d}(B_2)\right\} \leq \mathsf{d}(B_1)+ \mathsf{d}(B_2).
\end{align*}

\section{Quantum realization in the simplest scenario} \label{sec:appendix1}
    Consider the so called simplest scenarios $\mathbb{B}_{\rm si}$ from Def.~\ref{def: Simplest scenario}. In this scenario there are no operational equivalences between the measurement procedures $\mathbb{M}:= \{M_1,M_2\}$. A possible quantum contextual realization within the simplest scenario was presented in Refs.~\cite{Pusey2018simplest, spekkens2009preparation}, where we simply consider that $M_1,M_2$ are given by the POVMs, $M_1=\frac{1}{\sqrt{2}}(\sigma_X + \sigma_Z),M_2=\frac{1}{\sqrt{2}}(\sigma_X - \sigma_Z)$. For the preparations we can simply take, as is usual, the preparations \eqref{preparation 00}-\eqref{preparation --}, such that this realization gives rise immediately to the correct operational equivalences for preparation procedures from $\mathbb{B}_{\rm si}$. For a sufficiently large number of repeated procedures, we get the following probabilistic data-table ( we omit probabilities corresponding to $1-p$ events).

\begin{table}[H]
\centering
\begin{tabular}{|c|c|c|c|c|}
\hline
 & $\rho^1$ & $\rho^2$ & $\rho^3$ & $\rho^4$  \\
\hline
$E_{1}^{1}$ & $0.1464$ & $0.8535$ & $0.1464$ & $0.8535$ \\
\hline 
$E_{1}^{2}$ & $0.8535$ & $0.1464$ & $0.1464$ & $0.8535$ \\
\hline
\end{tabular}
\caption{Data-table for the final statistics obtained by quantum predictions.}
\label{tab: behavior numbers}
\end{table}

Since we have a full characterization of the noncontextual polytope for such scenario~\cite{schmid2018all} we notice that Table~\ref{tab: behavior numbers} is \textit{contextual} because it violates Ineq.~\eqref{counter example equation}.
\begin{align*}
    p_{12}+p_{24}-p_{22}-p_{13} = 0.8535+0.8525-0.1464-0.1464 = 1.4132 > 1,
\end{align*}

\section{Permutations of contextual vertices}\label{appenidx: symmetry}

In this Appendix we prove Lemma~\ref{lemma:general violations}. There are three partial results we need: 

\begin{enumerate}
    \item The set of vertices of the polytope of all behaviors $\mathbb{B}_{\rm si}$, is the set of deterministic assignments that respect the operational equivalences associated with the preparation procedures. Therefore there is always a permutation $T_{v\to w}$ between any two vertices from $\mathbb{B}_{\rm si}$. Since any permutation is also a free operation, we need to show that $T_{v \to w}$ is a special kind of free operations, one that satisfy,
    \begin{equation*}
        B \in C(\mathbb{B})\setminus NC(\mathbb{B}) \Leftrightarrow T_{v\ \to w}(B) \in C(\mathbb{B})\setminus NC(\mathbb{B})
    \end{equation*}
    This will be true whenever $T_{v \to w}$ is a permutation of elements in the behavior $B \in \mathbb{B}_{\rm si}$. 
    \item To do so, we will use the contextual measure $\mathsf{d} : \mathbb{B} \to \mathbb{R}_+$. We need to show that 
    \begin{equation}
        B \in C(\mathbb{B})\setminus NC(\mathbb{B}) \implies \mathsf{d}(B) > 0.
    \end{equation}
    \item We need to show that for any quantum contextual behavior $B \in \mathbb{B}_{\rm si}^{\boxplus n}$, $n \geq 1$, there always exists some transformation $T_{v \to w}$ between contextual vertices of $\mathbb{B}_{\rm si}^{\boxplus n}$ that maintains this behavior inside the set $\mathbb{B}_{\rm si}^{\boxplus n}$.  This will be true for  $ \mathbb{B}_{\rm si} ^{\boxplus n}$, because for any $n \geq 1$, some of the features regarding the polytope structure of the simplest scenario will remain.  
\end{enumerate}

We then proceed with the demonstration of these lemmas.

\begin{lemma}\label{d measures contextuality indeed}
Let $\mathbb{B} := (|J|,|I|,|K|,\mathbb{E}_{\mathbb{M}},\mathbb{E}_{\mathbb{P}})$ be any finitely defined prepare-and-measure scenario. Let $\mathsf{d}: \mathbb{B} \to \mathbb{R}_+$ be the $l_1$-contextuality monotone defined in Lemma~\ref{lemma:Compositionl1-distance}. Then, the following holds:
 \begin{equation}
        B \in C(\mathbb{B})\setminus NC(\mathbb{B}) \implies \mathsf{d}(B) > 0.
    \end{equation}
\end{lemma}

\begin{dem} 
Let $B \in C(\mathbb{B})\setminus NC(\mathbb{B})$. Then, $\mathsf{d}(B) = 0$ implies that,
\begin{align*}
    &\min_{B^* \in NC(\mathbb{B})}\max_{i,j}\sum_k \vert p(k|i,j) - p^*(k|i,j) \vert = 0 \implies \\
    &\exists \,\, B^* \in NC(\mathbb{B}), \max_{i,j}\sum_k \vert p(k|i,j) - p^*(k|i,j) \vert = 0,  \implies\\
    &\forall i\in I, \forall j\in J, \sum_k |p(k|i,j) - p^*(k | i,j) | = 0 \implies \\
    & \forall i\in I, \forall j \in J, \forall k \in K, \vert p(k|i,j) - p^*(k | i,j)  \vert = 0
\end{align*}
which implies that $B = B^*$, contradiction. Therefore, we must have that $\mathsf{d}(B) > 0$. \newline 
\end{dem}

\begin{lemma}
Let $\mathbb{B} := (|J|,|I|,|K|,\mathbb{E}_{\mathbb{M}},\mathbb{E}_{\mathbb{P}})$ be any finitely defined prepare-and-measure scenario. Let $P \in \mathcal{F}$ defined such that, for all $i \in I$, $q_O^i: K \to \tilde{K} = K$, $q_M : \tilde{I}=I \to I$, $q_P: \tilde{J}=J \to J$ are permutation matrices.  Then, 
\begin{equation}\label{completely free}
    B \in C(\mathbb{B})\setminus NC(\mathbb{B}) \Leftrightarrow P(B) \in C(\mathbb{B})\setminus NC(\mathbb{B}).
\end{equation}
Whenever $P \in \mathcal{F}$ is defined by permutation matrices we call $P$ a free permutation. If a free operation satisfy equation \eqref{completely free} we will refer to this operation as a completely free operation. 
\end{lemma}

\begin{dem}
 We want to show that free permutations are completely free transformations. Let $B \in \mathbb{B}$, and let $P \in \mathcal{F}$ be a free permutation. Since $P \in \mathcal{F}$ we immediately have that 
 
 \begin{equation*}
      B \in C(\mathbb{B})\setminus NC(\mathbb{B}) \Leftarrow P(B) \in C(\mathbb{B})\setminus NC(\mathbb{B}).
 \end{equation*}
 
 Therefore, we need to show that the arrow $\Rightarrow$ holds. Note that the action of $P$ is simply to rearrange the labels. We denote this by using the tilde notation, 

\begin{equation*}
    B := \left(p(k|i,j)\right)_{k \in K, i\in I, j\in J} \mapsto P(B) = \left(p(\Tilde{k} | \Tilde{i}, \Tilde{j})\right)_{\Tilde{k}\in K, \Tilde{i} \in I, \Tilde{j} \in J}.
\end{equation*}

This implies that, using the $l_1$-monotone, 

\begin{align*}
    \mathsf{d}(B) = \min_{B^* \in NC(\mathbb{B})}\max_{i\in I,j\in J}\sum_{k \in K} \vert p(k|i,j) - p^*(k|i,j) \vert = \min_{B^* \in NC(\mathbb{B})}\max_{\tilde{i}\in I,\tilde{j}\in J}\sum_{\tilde{k} \in K} \vert p(\tilde{k}|\tilde{i},\tilde{j}) - p^*(\tilde{k}|\tilde{i},\tilde{j}) \vert = \mathsf{d}(P(B)).
\end{align*}
We can then conclude that, 
\begin{equation}
    B \in C(\mathbb{B})\setminus NC(\mathbb{B}) \stackrel{\text{Lemma } \ref{d measures contextuality indeed}}{\implies} \mathsf{d}(B) > 0 \implies \mathsf{d}(P(B)) > 0 \implies P(B) \in C(\mathbb{B})\setminus NC(\mathbb{B})
\end{equation}
whenever $P$ is a free operation constructed from permutation matrices.

\end{dem}

\begin{lemma}
Let $\mathbb{B}_{\rm si}$ be the simplest scenario, as defined in Def.~\ref{def: Simplest scenario}. The noncontextual polytope $NC(\mathbb{B}_{\rm si})$ has eight facet-defining noncontextuality inequalities, that we associate with affine-linear functionals $\{h_i\}_{i=1}^8$, with $h_i : \mathbb{R}^{8} \to \mathbb{R}$. Then, for each affine-functional $h_i$ there is one, and only one vertex $B_v \in V(\mathbb{B}_{\rm si})$, for $V(\mathbb{B}_{\rm si})$ the set of vertices of the convex polytope $\mathbb{B}_{\rm si}$, that violates the noncontextuality inequality associated with $h_i(B) \leq 0$. 
\end{lemma}

\begin{dem}
Writing the functionals $h_i(B) \leq 0$ is equivalent to write the noncontextuality inequalities given by equations \eqref{simplest equations 1}-\eqref{simplest equations 9}. The polytope $NC(\mathbb{B}_{\rm si})$ has many other noncontextuality inequalities but these constitute the non-trivial tight noncontextuality inequalities. The lemma is proven then by construction: Table~\ref{tab: all behaviors vertices contextual} has all elements of $V(\mathbb{B}_{\rm si})$, and we highlighted the contextual vertices. $B \equiv \left(\begin{matrix}p_{11},p_{12},p_{13},p_{14}\\p_{21},p_{22},p_{23},p_{24}\end{matrix}\right)$.

\begin{table}[H]
\centering
\scalebox{1}{
\begin{tabular}{|c|c|c|c|c|c|c|}
\hline
 & $(1,1,1,1)$ & $(1,0,0,1)$ & $(1,0,1,0)$ & $(0,1,0,1)$  & $(0,1,1,0)$ & $(0,0,0,0)$\\
\hline
$(1,1,1,1)$ & $\left(\begin{matrix}1,1,1,1\\1,1,1,1\end{matrix}\right)$ & $\left(\begin{matrix}1,0,0,1\\1,1,1,1\end{matrix}\right)$ & $\left(\begin{matrix}1,0,1,0\\1,1,1,1\end{matrix}\right)$ & $\left(\begin{matrix}0,1,0,1\\1,1,1,1\end{matrix}\right)$  & $\left(\begin{matrix}0,1,1,0\\1,1,1,1\end{matrix}\right)$ & $\left(\begin{matrix}0,0,0,0\\1,1,1,1\end{matrix}\right)$\\
\hline
$(1,0,0,1)$ & $\left(\begin{matrix}1,1,1,1\\1,0,0,1\end{matrix}\right)$ & $\left(\begin{matrix}1,0,0,1\\1,0,0,1\end{matrix}\right)$ & $\textcolor{orange}{\left(\begin{matrix}1,0,1,0\\1,0,0,1\end{matrix}\right)}$ & $\textcolor{orange}{\left(\begin{matrix}0,1,0,1\\1,0,0,1\end{matrix}\right)}$  & $\left(\begin{matrix}0,1,1,0\\1,0,0,1\end{matrix}\right)$ & $\left(\begin{matrix}0,0,0,0\\1,0,0,1\end{matrix}\right)$\\
\hline
$(1,0,1,0)$ & $\left(\begin{matrix}1,1,1,1\\1,0,1,0\end{matrix}\right)$ & $\textcolor{orange}{\left(\begin{matrix}1,0,0,1\\1,0,1,0\end{matrix}\right)}$ & $\left(\begin{matrix}1,0,1,0\\1,0,1,0\end{matrix}\right)$ & $\left(\begin{matrix}0,1,0,1\\1,0,1,0\end{matrix}\right)$  & $\textcolor{orange}{\left(\begin{matrix}0,1,1,0\\1,0,1,0\end{matrix}\right)}$ & $\left(\begin{matrix}0,0,0,0\\1,0,1,0\end{matrix}\right)$\\
\hline
$(0,1,0,1)$ & $\left(\begin{matrix}1,1,1,1\\0,1,0,1\end{matrix}\right)$ & $\textcolor{orange}{\left(\begin{matrix}1,0,0,1\\0,1,0,1\end{matrix}\right)}$ & $\left(\begin{matrix}1,0,1,0\\0,1,0,1\end{matrix}\right)$ & $\left(\begin{matrix}0,1,0,1\\0,1,0,1\end{matrix}\right)$  & $\textcolor{orange}{\left(\begin{matrix}0,1,1,0\\0,1,0,1\end{matrix}\right)}$ & $\left(\begin{matrix}0,0,0,0\\0,1,0,1\end{matrix}\right)$\\
\hline
$(0,1,1,0)$ & $\left(\begin{matrix}1,1,1,1\\0,1,1,0\end{matrix}\right)$ & $\left(\begin{matrix}1,0,0,1\\0,1,1,0\end{matrix}\right)$ & $\textcolor{orange}{\left(\begin{matrix}1,0,1,0\\0,1,1,0\end{matrix}\right)}$ & $\textcolor{orange}{\left(\begin{matrix}0,1,0,1\\0,1,1,0\end{matrix}\right)}$  & $\left(\begin{matrix}0,1,1,0\\0,1,1,0\end{matrix}\right)$ & $\left(\begin{matrix}0,0,0,0\\0,1,1,0\end{matrix}\right)$\\
\hline
$(0,0,0,0)$ & $\left(\begin{matrix}1,1,1,1\\0,0,0,0\end{matrix}\right)$ & $\left(\begin{matrix}1,0,0,1\\0,0,0,0\end{matrix}\right)$ & $\left(\begin{matrix}1,0,1,0\\0,0,0,0\end{matrix}\right)$ & $\left(\begin{matrix}0,1,0,1\\0,0,0,0\end{matrix}\right)$  & $\left(\begin{matrix}0,1,1,0\\0,0,0,0\end{matrix}\right)$ & $\left(\begin{matrix}0,0,0,0\\0,0,0,0\end{matrix}\right)$\\
\hline
\end{tabular}
}
\caption{Table of deterministic vertices from $C(\mathbb{B}_{\rm si})$. We highlight the contextual vertices for the simplest scenario. Each of these vertices violate one of the tight noncontextuality inequalities defined by \eqref{simplest equations 1}-\eqref{simplest equations 9}. Note that for this scenario noncontextual vertices of $NC(\mathbb{B}_{\rm si})$ need not be deterministic.}
\label{tab: all behaviors vertices contextual}
\end{table}

Each one of these vertices violate one, and only one noncontextuality inequality.

\begin{align*}
    &h_1(B) = p_{12}+p_{22}-p_{14}-p_{23} - 1 \leq 0 \to h_1\left(\textcolor{orange}{\left(\begin{matrix}0,1,1,0\\0,1,0,1\end{matrix}\right)}\right) > 0\\
    &h_2(B) = p_{12}+p_{22}-p_{13}-p_{24} - 1 \leq 0 \to h_2\left(\textcolor{orange}{\left(\begin{matrix}0,1,0,1\\0,1,1,0\end{matrix}\right)}\right) > 0\\
    &h_3(B) = p_{22}+p_{13}-p_{12}-p_{24} - 1 \leq 0 \to h_3\left(\textcolor{orange}{\left(\begin{matrix}1,0,1,0\\0,1,1,0\end{matrix}\right)}\right) > 0\\
    &h_4(B) = p_{12}+p_{23}-p_{22}-p_{14} - 1 \leq 0 \to h_4\left(\textcolor{orange}{\left(\begin{matrix}0,1,1,0\\1,0,1,0\end{matrix}\right)}\right) > 0\\
    &h_5(B) = p_{22}+p_{14}-p_{12}-p_{23} - 1 \leq 0 \to h_5\left(\textcolor{orange}{\left(\begin{matrix}1,0,0,1\\0,1,0,1\end{matrix}\right)}\right) > 0\\
    &h_6(B) = p_{23}+p_{14}-p_{12}-p_{22} - 1 \leq 0 \to h_6\left(\textcolor{orange}{\left(\begin{matrix}1,0,0,1\\1,0,1,0\end{matrix}\right)}\right) > 0\\
    &h_7(B) = p_{12}+p_{24}-p_{22}-p_{13} - 1 \leq 0 \to h_7\left(\textcolor{orange}{\left(\begin{matrix}0,1,0,1\\1,0,0,1\end{matrix}\right)}\right) > 0\\
    &h_8(B) = p_{13}+p_{24}-p_{22}-p_{12} - 1 \leq 0 \to h_8\left(\textcolor{orange}{\left(\begin{matrix}1,0,1,0\\1,0,0,1\end{matrix}\right)}\right) > 0
\end{align*}

In terms of the polytope structure, each contextual behavior $B \in \mathbb{B}_{\rm si}$ violates at least one inequality, and therefore for some $h \in \{h_i\}_{i=1}^8$ and some $B_v \in V(\mathbb{B})$, we have that both $h(B) > 0$ and $h(B_v)>0$. 
\end{dem}

\begin{lemma}
Whenever $B_v,B_w \in V(\mathbb{B}_{\rm si})$ are contextual vertices, there exists a free permutation $T_{v \to w}$ satisfying the following: 
\begin{equation}
    \forall B \in \mathbb{B}_{\rm si} \implies T_{v \to w}(B) \in \mathbb{B}_{\rm si}.
\end{equation}
Since free permutations are completely free operations, we have that there exists a quantum contextual behavior violating all noncontextuality inequalities of $\mathbb{B}_{\rm si}$. 
\end{lemma}

\begin{dem}
Let $\Gamma(\mathbb{B}_{\rm si})$ be the symmetry group of the polytope $\mathbb{B}_{\rm si} \subset \mathbb{R}^8 \subset \mathbb{R}^{16}$. For any convex polytope $\mathbb{B}_{\rm si}$, the symmetry group $\Gamma(\mathbb{B}_{\rm si})$ is finite. To each symmetry $\alpha \in \Gamma(\mathbb{B}_{\rm si})$ we associate a (faithful) representation $T: \Gamma(\mathbb{B}_{\rm si}) \to \text{Aut}(\mathbb{R}^8)$. Define $c(\mathbb{B}_{\rm si}) := \text{ConvHull}\left( V(\mathbb{B}_{\rm si})\setminus V(NC(\mathbb{B}_{\rm si}))\right)$. Then, we have that the polytope $c(\mathbb{B}_{\rm si})$ is vertex-transitive, i.e., there exists an element of $\Gamma(\mathbb{B}_{\rm si})$ that sends any vertex of $c(\mathbb{B}_{\rm si})$ to any other vertex.  To see this we consider the group elements $\alpha,\beta,\gamma, \delta$ defined by their representation matrices as, for any point $\left(\begin{matrix}p_{11},p_{12},p_{13},p_{14}\\p_{21},p_{22},p_{23},p_{24}\end{matrix}\right) \in \mathbb{B}_{\rm si}$: 

\begin{table}[H]
\centering
\begin{tabular}{|c|c|c|c|c|}
\hline
 Element & Result  & $q_O$ & $q_M$ & $q_P$  \\
\hline
$\alpha$ & $\left(\begin{matrix}p_{12},p_{11},p_{13},p_{14}\\p_{22},p_{21},p_{23},p_{24}\end{matrix}\right)$ & $\left(\begin{matrix}1&0\\0&1\end{matrix}\right)$ & $\left(\begin{matrix}1&0\\0&1\end{matrix}\right)$ & $\left(\begin{matrix}0&1&0&0\\1&0&0&0\\0&0&1&0\\0&0&0&1\end{matrix}\right)$ \\
\hline 
$\beta$ & $\left(\begin{matrix}p_{11},p_{12},p_{14},p_{13}\\p_{21},p_{22},p_{24},p_{23}\end{matrix}\right)$ & $\left(\begin{matrix}1&0\\0&1\end{matrix}\right)$ & $\left(\begin{matrix}1&0\\0&1\end{matrix}\right)$ & $\left(\begin{matrix}1&0&0&0\\0&1&0&0\\0&0&0&1\\0&0&1&0\end{matrix}\right)$ \\
\hline
$\gamma$ & $\left(\begin{matrix}p_{21},p_{22},p_{23},p_{24}\\p_{11},p_{12},p_{13},p_{14}\end{matrix}\right)$ & $\left(\begin{matrix}1&0\\0&1\end{matrix}\right)$ & $\left(\begin{matrix}0&1\\1&0\end{matrix}\right)$ & $\left(\begin{matrix}1&0&0&0\\0&1&0&0\\0&0&1&0\\0&0&0&1\end{matrix}\right)$ \\
\hline
$\delta$ & $\left(\begin{matrix}p_{13},p_{14},p_{11},p_{12}\\p_{23},p_{24},p_{21},p_{22}\end{matrix}\right)$ & $\left(\begin{matrix}1&0\\0&1\end{matrix}\right)$ & $\left(\begin{matrix}1&0\\0&1\end{matrix}\right)$ & $\left(\begin{matrix}0&0&1&0\\0&0&0&1\\1&0&0&0\\0&1&0&0\end{matrix}\right)$ \\
\hline
\end{tabular}
\caption{Free permutations acting over $\mathbb{B}_{\rm si}$. Each one of these operations is such that they leave the polytope $\mathbb{B}_{\rm si}$ invariant. They also leave $c(\mathbb{B})_{si}$ invariant, meaning that each vertex is sent to another vertex in the polytope. }
\label{tab: free operations actions}
\end{table}
The free operations that we associate with the group elements $\alpha,\beta,\gamma,\delta \in \Gamma(\mathbb{B}_{\rm si})$, present in Table~\ref{tab: free operations actions} they constitute free permutations. Each one of these operations is such that they constitute a proof that $c(\mathbb{B}_{\rm si})$ is a vertex-transitive convex polytope. To prove so we simply notice that we can make a graph, such that each contextual point is a vertex of the graph, and each edge is a free permutation between the vertices. Figure~\ref{fig:transformation between contextual vertices} we construct this graph; since this graph is completely connected, there is always a symmetry (free permutation) between any two vertices\footnote{The set $\mathcal{F}$ of free operation is closed under composition. } and we conclude that $c(\mathbb{B}_{\rm si})$ is a vertex-transitive convex polytope.
\begin{figure}[H]
    \centering
    \scalebox{2}{
    \begin{tikzpicture}
    \draw[fill] (0,0) circle[radius=0.05];
    \draw[fill] (0,-1) circle[radius=0.05];
    \draw[fill] (1,-2) circle[radius=0.05];
    \draw[fill] (2,-2) circle[radius=0.05];
    \draw[fill] (3,-1) circle[radius=0.05];
    \draw[fill] (3,0) circle[radius=0.05];
    \draw[fill] (2,1) circle[radius=0.05];
    \draw[fill] (1,1) circle[radius=0.05];
    \draw[blue] (0,0) -- (1,1);
    \draw[blue] (0,-1) -- (2,1);
    \draw[blue] (1,-2) -- (3,0);
    \draw[blue] (2,-2) -- (3,-1);
    \draw[magenta] (0,0) -- (2,-2);
    \draw[magenta] (1,1) -- (3,-1);
    \draw[black] (0,-1) -- (1,-2);
    \draw[black] (2,1) -- (3,0);
    \node[blue] at (0.35,0.65) {\scriptsize $\gamma$};
    \node[blue] at (1.55,0.75) {\scriptsize $\gamma$};
    \node[blue] at (2.15,-0.65) {\scriptsize $\gamma$};
    \node[blue] at (2.35,-1.45) {\scriptsize $\gamma$};
    \draw[teal] (0,0) -- (3,0);
    \draw[teal] (0,-1) -- (3,-1);
    \draw[teal] (1,1) -- (1,-2);
    \draw[teal] (2,1) -- (2,-2);
    \node[magenta] at (0.3,-0.5) {\scriptsize $\alpha$};
    \node[magenta] at (2.95,-0.5) {\scriptsize $\alpha$};
    \node[black] at (2.75,0.5) {\scriptsize $\beta$};
    \node[black] at (0.3,-1.5) {\scriptsize $\beta$};
    \node[teal] at (0.85,0.35) {\scriptsize $\delta$};
    \node[teal] at (1.85,0.35) {\scriptsize $\delta$};
    \node[teal] at (2.55,0.25) {\scriptsize $\delta$};
    \node[teal] at (2.55,-0.85) {\scriptsize $\delta$};
    \end{tikzpicture}
    }
    \caption{Graph of transformations between contextual vertices.}
    \label{fig:transformation between contextual vertices}
\end{figure}
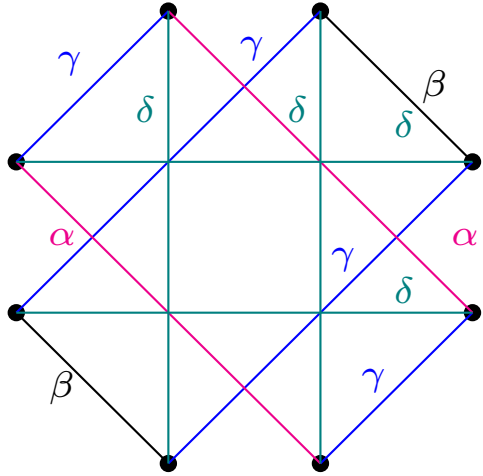

Since there exists a bijection between nontrivial violations given by equations \eqref{simplest equations 1}-\eqref{simplest equations 9}, and contextual vertices in the scenario $\mathbb{B}_{\rm si}$ we have that if $B \in \mathbb{B}_{\rm si}$ violates an inequality associated with $B_v$ a contextual vertex, since $c(\mathbb{B}_{\rm si})$ is vertex-transitive there exists a symmetry $T_{v \to w}$, that is a free permutation that can be read from figure \ref{fig:transformation between contextual vertices}, such that $T_{v \to w}(B)$ violates any other noncontextuality inequality associated with any other contextual vertex $B_w \in V(c(\mathbb{B}_{\rm si}))$. Therefore, let $B^Q \in \mathbb{B}_{\rm si}$ be the  quantum contextual behavior given by Table~\ref{tab: behavior numbers}. There exists a free permutation that sends $B^Q$ to a region of the polytope $\mathbb{B}_{\rm si}$ that violates any other noncontextuality inequality.  

\end{dem}

Now that we know that this is a true feature of $\mathbb{B}_{\rm si}$, we can prove the result from the text (Lemma~\ref{lemma:general violations}), that we rewrite here. 

\begin{lemma}
For any scenario of the form  $\mathbb{B}:=\mathbb{B}_{\rm si}^{\boxplus n}$, $n\geq 1$, every tight noncontextuality inequality will have a quantum contextual behavior. 
\end{lemma}

\begin{dem}
Since we know that this is true for $n=1$ we can try to prove this by induction. Let the set of non-trivial noncontextuality inequalities for $\mathbb{B}_{\rm si}$ be defined by $H_{1} := \{h_i: h_i(B) \leq 0, i \in \{1,\dots,8\}\}$. For $\mathbb{B}$ and $\mathbb{B}_{\rm si}^{\boxplus (n-1)}$ we use similar definitions, denoting $H$ the set of functionals corresponding to inequalities of $NC(\mathbb{B})$, and $H_{n-1}$ for $NC(\mathbb{B}_{\rm si}^{\boxplus (n-1)})$. Since we know that $NC(\mathbb{B}) = NC(\mathbb{B}_{\rm si})\times NC(\mathbb{B}_{\rm si}^{\boxplus (n-1)})$ we have that $H$ is the set of linear functionals such that, for all $\mathbb{B} \ni B = B_1 \boxplus B_2$,  
\begin{equation}\label{property of facets in a product polytope}
H := \{h : h(B) = h_1(B_1)\leq 0, \text{for some }h_1 \in H_1\} \cup \{h : h(B) = h_2(B_2)\leq 0, \text{for some }h_2 \in H_{n-1}\}
\end{equation}

Therefore, as our hypothesis we suppose that there exists a quantum contextual behavior $B^Q$ such that, for all $h \in H_{n-1}$ we have that $h(B^Q) > 0$. Since, for $H_1$ we know that this is also true and that $H$ is given by equation \eqref{property of facets in a product polytope}, we have that for every $h \in H$ there exists some $B^Q \in \mathbb{B}$ such that $h(B^Q) > 0$. We conclude that this must be true for all $n \geq 1$. 
\end{dem}

Notice that equation \eqref{property of facets in a product polytope} is another way of stating that the number of facets in a convex polytope that is the product of two convex polytopes gets summed. This is clear by noticing that we can associate the sets $H$ to matrices that define the convex polytope via an H-representation, and by proving Lemma~\ref{lemma: product convex polytopes}, from where it is clear how the set of inequalities (convex-linear functions in the terminology of the sets $H$) is upgrated for the product polytope.

\end{appendix}

\end{document}